\documentclass{article}

\PassOptionsToPackage{numbers, compress, sort}{natbib}
\usepackage[preprint]{neurips_2025}

\usepackage[utf8]{inputenc} 
\usepackage[T1]{fontenc}    
\usepackage{hyperref}       
\usepackage{url}            
\usepackage{booktabs}       
\usepackage{amsfonts}       
\usepackage{nicefrac}       
\usepackage{microtype}      
\usepackage{xcolor}         

\usepackage{amsmath}
\usepackage{amssymb}
\usepackage{mathtools}
\usepackage{amsthm}

\usepackage[noend]{algorithmic}

\usepackage{shortcuts}
\usepackage{siunitx}

\usepackage[textsize=tiny]{todonotes}
\usetikzlibrary{plotmarks}
\newcommand\marksymbol[2]{\tikz[#2,scale=1.2]\pgfuseplotmark{#1};}

\title{Uncovering Bias Mechanisms in Observational Studies}

\author{%
  Ilker Demirel\thanks{Equal contribution.} \\
  MIT CSAIL\\
  \And
  Zeshan Hussain$^{*}$ \\
  MIT CSAIL \\
  \And
  Piersilvio De Bartolomeis \\
  ETH Zurich \\
  \And
  David Sontag \\
  MIT CSAIL \\
}

\begin{document}

\maketitle

\begin{abstract}
    Observational studies are a key resource for causal inference but are often affected by systematic biases. Prior work has focused mainly on detecting these biases, via sensitivity analyses and comparisons with randomized controlled trials, or mitigating them through debiasing techniques. However, there remains a lack of methodology for uncovering the underlying mechanisms driving these biases, {\em e.g.}, whether due to hidden confounding or selection of participants. In this work, we show that the relationship between bias magnitude and the predictive performance of nuisance function estimators (in the observational study) can help distinguish among common sources of causal bias. We validate our methodology through extensive synthetic experiments and a real-world case study, demonstrating its effectiveness in revealing the mechanisms behind observed biases. Our framework offers a new lens for understanding and characterizing bias in observational studies, with practical implications for improving causal inference.
\end{abstract}

\section{Introduction}
Observational data from electronic health records and insurance claims offer compelling advantages for causal inference. They incur less cost to collect and provide broader population representation than randomized controlled trials (RCT), allowing analyses in subgroups not supported in RCTs \cite{govc2019, nice2022uk}. However, observational studies (OS) are plagued by biases that undermine their reliability, such as hidden confounding or selection bias from loss to follow-up \cite{hernan2008observational, imbens2015causal, pearl2018book, lodi2019effect}.

In order to reliably utilize OS results, it is crucial for practitioners to have a nuanced understanding of the reasons behind the bias within a study. Doing so could allow one to improve the study design or causal analysis to address the source of bias. A good example is found in \citet{dagan2021bnt162b2}, where the authors iteratively refine their covariate adjustment set via calibration with negative controls to estimate the efficacy of COVID-19 vaccinations (see \Cref{sec:dags}). To tackle this issue, we develop methodology to distinguish between a set of causal bias mechanisms commonly found in OSes.

Our approach is related to two developing lines of work\footnote{We include an extended related work section in \Cref{eq:ext_rel_work}.}: i)~benchmarking OSes against RCTs \citep{hartman2015sample, kallus2018removing, forbes2020benchmarking, hussain2022falsification, demirel2024benchmarking, de2024detecting}, and ii) recovering the causal graph using observational or interventional data \citep{shimizu2006linear, shanmugam2015learning, heinze2018causal, glymour2019review, vowels2022d, choo2022verification, squires2023causal}. While the benchmarking literature focuses on spotting biases in OSes, it does not produce insights on the reasons underlying the bias. On the other hand, the structure learning literature proposes to specify the causal graph from observed data, which would uncover the exact bias mechanism. However, this task is often challenging and requires strong assumptions, many of which are unrealistic in practice. 

Rather than attempting to recover an exact causal graph--a task that may be too ambitious within practically defensible assumptions--we focus on distinguishing between specific families of causal graphs corresponding to different bias mechanisms (see \Cref{sec:dags}), affording practitioners actionable insights to improve their study design. This task is inherently more tractable, as it does not require recovering a single graph but instead groups together several graphs that exhibit equivalent bias patterns. To that end, we begin by estimating the bias function in the OS using data from a RCT conducted on a population supported in both datasets. We then examine how the bias varies across patients as a function of the performance of predictive models\footnote{termed nuisance function estimators in the double machine learning literature, see \citet{tsiatis2006semiparametric}.} fitted on the OS data. This relationship gives rise to empirical statistics that can effectively differentiate between distinct bias mechanisms.

\paragraph{Contributions} First, we establish a comprehensive taxonomy of common causal biases, each corresponding to a distinct class of graphs (\Cref{sec:dags}). We describe a generative model for the OS in the context of these graphs motivated by clinical decision-making (\Cref{sec:GMSCM}). Next, we demonstrate a relationship between the predictive performance of the nuisance functions and the bias function in the OS (\Cref{sec:BDVTV}). To quantify this relationship, we use the covariance between the prediction error and magnitude of the bias, which enables provable discrimination of different bias mechanisms (\Cref{sec:IBMCCV}). For practical implementation, we propose consistent estimators of the covariance (\Cref{sec:AUEC}). Finally, we validate our methodology through both synthetic experiments and real-world analysis using data from Women's Health Initiative \citep{study1998design} (Sections~\ref{sec:synthetic} and \ref{sec:real-world}). 
\section{Notation and Background} 
\label{sec:back_and_mot}
Let $A$ denote a treatment action, $X$ the set of {\em measured} patient covariates at baseline, and $Y$ the observed outcome of interest. We denote by $Y^a$ the {\em potential} outcome under $A=a$. For each patient $i$, we observe only one of their {\em potential} outcomes, that is, $Y_i = Y^{A_i}$. 

We assume access to patient-level data from an RCT and an OS, and use $R = 1$ and $R = 0$ to represent the underlying RCT and  OS {\em populations}, respectively. We use $S$ to denote whether a patient was {\em selected} into the study cohort for analysis. For instance, a patient may be excluded from the analysis in an RCT if they did not adhere to their treatment assignment ({\em i.e.,} $R_i=1$ and $S_i=0$). From a large insurance claims dataset, a patient maybe selected into the OS cohort when emulating an RCT if they meet the eligibility criteria ({\em i.e.,} $R_i=0$ and $S_i=1$). We assume that $X$ is available for all patients, and that $A$ and $Y$ are available for those selected into the analysis ($S_i = 1$).

Finally, we let $U$ denote the set of {\em unmeasured} covariates that can influence the downstream variables $S$, $A$, and $Y^a$ in the OS. Such omitted variables are the reasons behind many common causal biases in observational studies, which are described in \Cref{sec:dags}. 

The causal estimand we focus on is the conditional average treatment effect (CATE), defined as
\begin{align} 
    &\text{CATE}_{\text{rct}} (x,a) \coloneqq \E[Y^a \mid X=x, R=1]. \label{eq:cate_def_rct} \\
    &\text{CATE}_{\text{os}} (x,a) \coloneqq \E[Y^a \mid X=x, R=0]. \label{eq:cate_def_os}
\end{align}
Estimating the CATE is challenging when unobserved covariates, $U$, influence treatment assignment ($A$), outcome ($Y$), and selection into the study cohort ($S$). We list below necessary conditions to identify the CATE in a population, which are often satisfied in RCTs, but can be violated in OSes.
\begin{assumption}[Internal validity of RCT] \label{asm:rct_iv}
    The following hold in the RCT ($R=1$) for all $a \in \{0,1\}$. \textit{Ignorability of selection} -- $Y^a \indep S \mid X, R = 1$. \textit{Ignorability of treatment} -- $Y^a \indep A \mid X, S, R = 1$. \textit{Positivity} -- $P(S=1, A=a \mid X, R = 1) > 0$.
\end{assumption}
Under \Cref{asm:rct_iv}, $\text{CATE}_{\text{rct}} (x,a)$ in \eqref{eq:cate_def_rct} can be identified as 
\begin{align}
    g_a (X) \coloneqq \E[Y \mid X, R=1, S=1, A=a] = \text{CATE}_{\text{rct}} (x,a) , \label{eq:gdef}
\end{align}
which can be estimated from the RCT cohort. The outcome model in the OS population is defined similarly below.
\begin{equation} \label{eq:fdef}
    f_a(X) \coloneqq \E[Y \mid X, R=0, S=1, A=a] \neq \text{CATE}_{\text{os}} (x,a),
\end{equation}
which in general is not a valid identification of $\text{CATE}_{\text{os}} (x,a)$. We define the bias function in the OS,
\begin{align}
    b_a(x) 
    &\coloneqq g_a(x) - f_a(x) = \underbrace{\text{CATE}_{\text{rct}} (x,a) - \text{CATE}_{\text{os}} (x,a)}_{\eqqcolon~\text{Transportability bias}} + \underbrace{\text{CATE}_{\text{os}} (x,a) - f_a(x) }_{\eqqcolon~\text{Internal bias in the OS}}. \label{eq:bias}
\end{align}
There could be various mechanisms underlying the bias, which cannot be inferred from $b_a (x)$ alone. To fill this gap, we develop methodology to uncover which mechanism could be driving the bias. To that end, we first give a taxonomy of common bias mechanisms in the next section.

\section{A Taxonomy of Common Causal Biases} \label{sec:dags}
\begin{figure}[t]
    \centering
    \hfill
    \begin{subfigure}{0.24\textwidth}
        \centering
        \begin{tikzpicture}[scale=0.7]
            
            \node[state] (x) {$X$};
            \node[state,red,dashed] (u) [right =of x, xshift=-0.6cm] {$U$};
            
            \node[state] (a) [below left =of x, yshift=0.4cm, xshift=0.8cm] {$A$};
            \node[state] (y) [below right=of x, yshift=0.4cm, xshift=-0.8cm] {$Y$};
            \node[state] (s) [right =of y, xshift=-0.6cm] {$S$};

            \node[anchor=south west] at (-1.6,0.6) {\small \textcolor{black}{$P(U|R\!=\!0)\!\neq\!P(U|R\!=\!1)$}};

            \path (x) edge  (a);
            \path (x) edge  (y);
            \path (x) edge  (s);
            \path (a) edge  (y);

            \path (u) edge[red]  (y);
            \path (a) edge[white,bend right]  (s);
        \end{tikzpicture}
        \caption{Transportability Bias.}
        \label{fig:transportability_bias}
    \end{subfigure}%
    \hfill
    \begin{subfigure}{0.24\textwidth}
        \centering
        \begin{tikzpicture}[scale=0.7]
            
            \node[state] (x) {$X$};
            \node[state,red,dashed] (u) [right =of x, xshift=-0.6cm] {$U$};
            
            \node[state] (a) [below left =of x, yshift=0.4cm, xshift=0.8cm] {$A$};
            \node[state] (y) [below right=of x, yshift=0.4cm, xshift=-0.8cm] {$Y$};
            \node[state] (s) [right =of y, xshift=-0.6cm] {$S$};

            \path (x) edge  (a);
            \path (x) edge  (y);
            \path (x) edge  (s);
            \path (a) edge  (y);

            \path (u) edge[red]  (y);
            \path (u) edge[red]  (a);
            \path (a) edge[white,bend right]  (s);
        \end{tikzpicture}
        \caption{Confounding Bias}
        \label{fig:confounding_bias}
    \end{subfigure}%
    \hfill
    \begin{subfigure}{0.24\textwidth}
        \centering
        \begin{tikzpicture}[scale=0.7]
            
            \node[state] (x) {$X$};
            \node[state,red,dashed] (u) [right =of x, xshift=-0.6cm] {$U$};
            
            \node[state] (a) [below left =of x, yshift=0.4cm, xshift=0.8cm] {$A$};
            \node[state] (y) [below right=of x, yshift=0.4cm, xshift=-0.8cm] {$Y$};
            \node[state] (s) [right =of y, xshift=-0.6cm] {$S$};

            \path (x) edge  (a);
            \path (x) edge  (y);
            \path (x) edge  (s);
            \path (a) edge  (y);

            \path (u) edge[red]  (y);
            \path (u) edge[red]  (s);
            \path (a) edge[white,bend right]  (s);
        \end{tikzpicture}
        \caption{Selection Bias 1}
        \label{fig:selection_bias_1}
    \end{subfigure}%
    \hfill
    \begin{subfigure}{0.24\textwidth}
        \centering
            \begin{tikzpicture}[scale=0.7]
            
            \node[state] (x) {$X$};
            
            \node[state] (a) [below left =of x, yshift=0.4cm, xshift=0.8cm] {$A$};
            \node[state] (y) [below right=of x, yshift=0.4cm, xshift=-0.8cm] {$Y$};
            \node[state] (s) [right =of y, xshift=-0.6cm] {$S$};

            \path (x) edge  (a);
            \path (x) edge  (y);
            \path (x) edge  (s);
            \path (a) edge  (y);

            \path (a) edge[white,bend right]  (s);

            \path (y) edge[red]  (s);
            \path (a) edge[red,bend right]  (s);
        \end{tikzpicture}
        \caption{Selection Bias 2}
        \label{fig:selection_bias_2}
    \end{subfigure}%
    \caption{Graphs of common biases in observational studies. $U$ denotes unmeasured covariates.}
    \label{fig:bias_typs}
    \vspace{-10pt}
\end{figure}
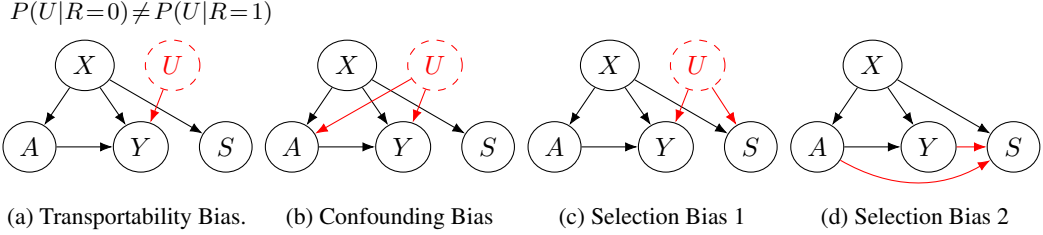
We focus on four types of common biases depicted in \Cref{fig:bias_typs}. Directed edges imply a causal relation, {\em e.g.}, intervening on $X$ directly induces a change in $Y$. We bucket the bias mechanisms into two categories. The first category involves an {\em unobserved} covariate, $U$, affecting a subset of downstream variables. The second category involves \enquote{collider} bias, where conditioning on a collider variable induces a spurious association between the treatment and outcome. Each of these biases represents a family of graphs ({\em e.g.}, see \Cref{sec:add-dags-t2sb}). We give an overview of these biases with real-world examples before developing methods to distinguish between them.
\subsection{Category 1: Effect of Unobserved Covariates}
\paragraph{Transportability Bias} Even when the OS has internal validity ($f_a(X) = \text{CATE}_{\text{os}} (x,a)$), bias may persist ($b_a(X) \neq 0$) due to poor transportability between the RCT and OS. Transportability bias arises when an unobserved covariate $U$ affects outcomes $Y^a$ and has different distributions across two populations (Figure \ref{fig:transportability_bias}). This effect can be observed as follows. Note that,
\begin{equation*}
    P(Y^a \mid X, R) =\sum\nolimits_{u} P(Y^a \mid X, U=u, R)  P(U=u \mid X, R).
\end{equation*}
When $P(U|X, R=0) \neq P(U | X, R=1)$, we have $\text{CATE}_{\text{os}} (x,a) \neq \text{CATE}_{\text{rct}} (x,a)$ and $b_a(X) \neq 0$ in general. 
For example, age could be an effect modifier, whereby a chemotherapy for breast cancer may have a larger effect on younger women (studied in the RCT) compared to older women (studied in the OS) \cite{ring2021bridging}. If age is not properly accounted for, there will be transportability bias.
\paragraph{Confounding Bias} Confounding bias is a core impediment to causal inference from observational data. When covariates $X$ miss factors that affect both the treatment assignment $A$ and outcomes $Y$ (Figure \ref{fig:confounding_bias}), the ``ignorability of treatment'' condition in \Cref{asm:rct_iv} is violated in the OS ($R=0$). This violation implies $f_a(X)\!\neq\!\text{CATE}_{\text{os}} (x,a)$, leading to non-zero bias $b_a(X)$.

To illustrate, consider COVID-19 vaccination studies: patients with higher socioeconomic status (SES) may be more likely to both receive vaccines and maintain preventive health behaviors. When observational data inadequately captures SES variations (unobserved confounder $U$), researchers risk introducing confounding bias that inflates vaccine effectiveness estimates \cite{dagan2021bnt162b2}.
\paragraph{Selection Bias Type 1} Selection bias emerges when the inclusion criteria lead to a sample that is unrepresentative of the underlying population. We examine two graphs capturing diverse selection mechanisms. The first type (\Cref{fig:selection_bias_1}) arises when an unmeasured covariate $U$ affects the selection $S$ and outcome $Y$, violating \enquote{ignorability of selection} in \Cref{asm:rct_iv} in the OS ($R=0$). 

Excluding patients who were lost-to-follow-up from the analysis is one way this type of bias can occur. As an example, \citet{hernan2004structural} give antiretroviral therapy for preventing AIDS in HIV-infected patients, where the true unmeasured immunosuppression level, $U$, affects both the AIDS risk ($Y$), and follow-up ($S$) through adverse side effects.

As another example, consider studies where participants opt-in to share their health data via a smart watch. Individuals with healthier lifestyles, such as a balanced diet, $U$, might be more likely to participate ($S=1$). Since these individuals typically experience lower heart failure (HF) rates, $Y$, the study may underestimate the interventions' benefits (e.g., regular jogging) for reducing HF rates in the general population.
\subsection{Category 2: Conditioning on Colliders}
\label{sec:sbt2-cc}
\paragraph{Selection Bias Type 2}
This bias (\Cref{fig:selection_bias_2}) arises when selection is affected by the treatment assignment and outcome, and it is particularly elusive \cite{holmberg2022collider,lipsky2022causal}. Initial analyses of the Women's Health Initiative OS showed protective effects of post-menopausal combined hormonal therapy against coronary heart disease and stroke \citep{barrett1998hormone}. However, subsequent analysis of the RCT component revealed increased risk \citep{manson2003estrogen}. The discrepancy stemmed from selective enrollment: patients who initiated treatment pre-enrollment and continued without early adverse events were preferentially selected ($A \rightarrow S$), while early events precluded selection into the treatment group ($Y \rightarrow S$), deflating event rates in the treatment cohort \citep{prentice2005combined}.
\section{Uncovering Bias Mechanisms via Alignment with Predictive Performance} \label{sec:ja_wnf}
In this section, we present our core methodology for distinguishing between the bias mechanisms outlined earlier and illustrated in \Cref{fig:bias_typs}, under clinically motivated data-generating assumptions. We begin with a clinical example that motivates the data-generating process and lay out the key assumptions in \Cref{sec:GMSCM}. Next, in \Cref{sec:BDVTV}, we develop statistics that can uniquely characterize the bias mechanisms under investigation and derive formal guarantees in Sections \ref{sec:IBMCCV} and \ref{sec:T2SB}. Finally, in \Cref{sec:AUEC}, we develop consistent estimators of those statistics.

Although the following sections introduce assumptions that facilitate sharp theoretical analysis, these conditions are not required for our methods to be applicable in practice. For analytical tractability, we assume a single binary unmeasured confounder $U$. However, core insights from our results extend naturally to more complex scenarios. Indeed, our synthetic experiments include settings where $U$ is continuous, and we find that the empirical results closely align with our theoretical predictions (see \Cref{app:moresyn}). Similarly, our real-world case study in \Cref{sec:real-world} likely involves multiple, continuous unmeasured confounders, yet our methodology provides clinically meaningful insights—highlighting its robustness and utility beyond the theoretical setting.

\subsection{Clinically Motivated Generative Model in Observational Data} \label{sec:GMSCM}

\begin{wrapfigure}{r}{0.5\textwidth}
\vspace{-1.15em} 
\begin{minipage}{0.48\textwidth}
  \hrule
  \hrule
  \vspace{.5em}
  \textbf{Algorithm 1:} Generative Model in the OS \\
  \vspace{-.6em}
  \hrule
  \begin{algorithmic}
        \STATE \textbf{Input:} $\orange{T} \in \{ S,A,Y^0, Y^1 \}$; Boolean \orange{$U$-bias}; \\ Bernoulli Parameter Distribution \orange{${\cal F} (p)$}.
        
        \hrulefill
        \STATE Let $\orange{p^{T}_{x, u}} \coloneqq P(T\!=\!1 | X\!=\!x, U\!=\!u, R\!=\!0)$
        \FOR{$x \in {\cal X}$}
            \STATE $p^{T}_{x, u=0} \sim {\cal F}(p)$
                \IF{$U$-bias}
                    \STATE $p^{T}_{x, u=1} \sim {\cal F}(p)$
                \ELSE
                    \STATE $p^{T}_{x, u=1} = p^{T}_{x, u=0}$
                \ENDIF
        \ENDFOR
  \end{algorithmic}
  \hrule
  \hrule
\end{minipage}
\vspace{-1.5em}
\end{wrapfigure}
We consider a binary treatment $A$, outcome $Y$, unmeasured covariate $U$, and a general covariate set $X \in {\cal X}$ which can include both categorical and continuous features. The generative model for the downstream variables in the OS is given in Algorithm 1. If $U$-bias is \texttt{False} for a $T \in \{S,A,Y^0,Y^1\}$, then $U$ has no residual effect on $T$ after conditioning on $X$.

When $U$-bias is \texttt{True}, $U$ has a residual effect on $T$, even after conditioning on $X$. However, the significance of this effect can vary for different values of $X$, as we discuss next.

\paragraph{Magnitude of $U$'s effect varies across $X$.} 
The influence of an unmeasured covariate $U$ on a downstream variable $T \in \{S,A,Y^0,Y^1\}$ varies across patient types $X=x$. Specifically, its effect is small when $p^T_{x,u=1} \approx p^T_{x,u=0}$, and large otherwise (see Algorithm~1). This general property, often overlooked in practice, forms the foundation of our methods as we discuss in detail in \Cref{sec:BDVTV}. We illustrate this property with an example from clinical practice.

\paragraph{Clinical diagnostic reasoning reveals \textit{contextual independencies}.} When a patient presents, the clinician follows a systematic approach, beginning with their history and symptoms, followed by a set of assessments and diagnostic tests. This process often involves evaluating a set of clinical findings and test results simultaneously. Based on these initial findings ({\em i.e.,} context $X$), distinct diagnostic pathways emerge, where the importance of subsequent clinical variables may shift dramatically.

For example, consider the evaluation of thyroid nodules. The initial assessment includes demographics, medical history, physical examination, and thyroid function tests. In a young male patient with normal thyroid function and a solitary thyroid nodule, a family history ($U$) of thyroid cancer is critical in guiding immediate fine-needle aspiration, while menopausal status is irrelevant. Conversely, in a post-menopausal female patient with the same nodule but a history of radiation exposure, the family history ($U$), becomes relatively less influential in the diagnosis, as they already are in the high-risk group requiring a biopsy. This demonstrates how patient demographics and history create distinct contexts that render certain clinical variables more or less relevant to the decision-making.

\begin{wrapfigure}{l}{0.53\linewidth}
    \vspace{-15pt}
    \centering
    \includegraphics[width=\linewidth]{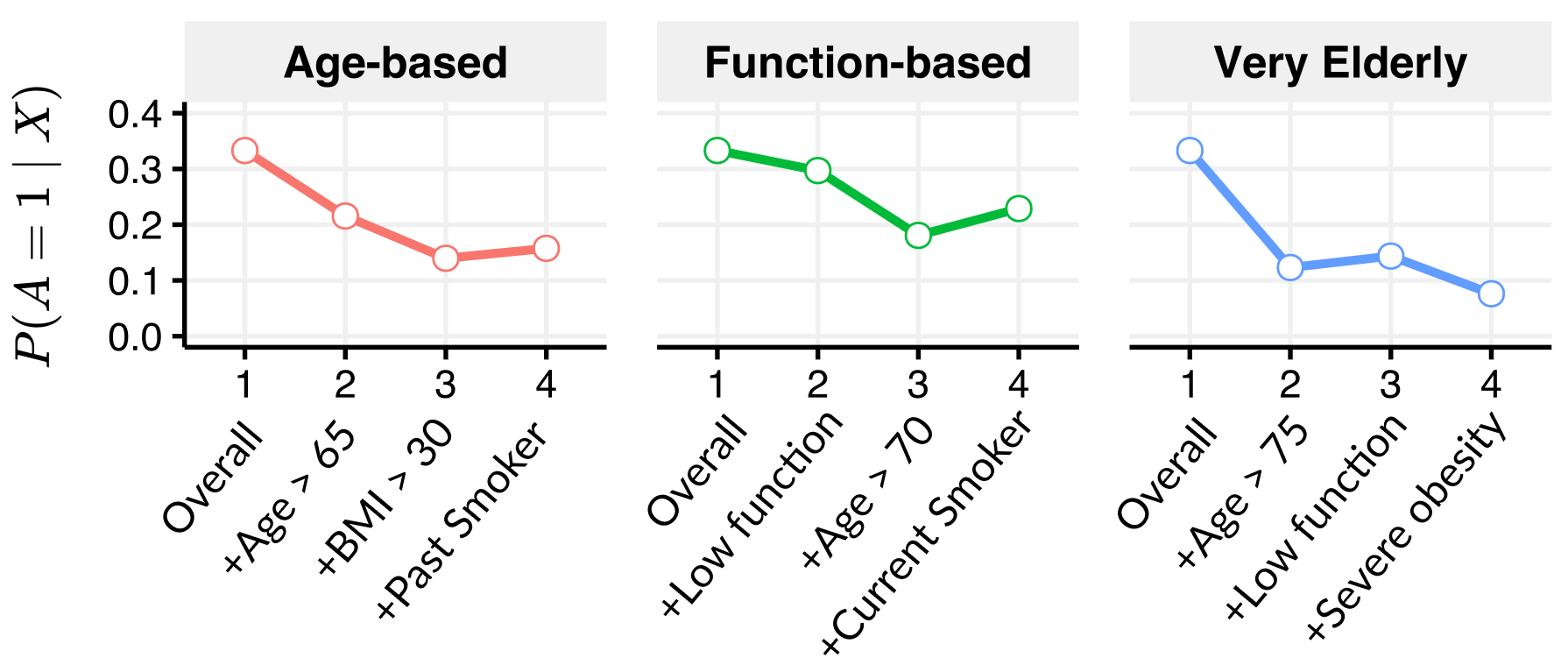}
    \caption{\small Evaluation of treatment assignment probability in three different subgroups in the WHI data \cite{study1998design}. As the conditioning set is expanded and the patient context is better specified, uncertainty in the decisions tend to decrease.}
    \label{fig:whi_support}
    \vspace{-10pt}
\end{wrapfigure}

\paragraph{Conditioning induces low uncertainty.}
The vignette above reveals another key property of clinical decision-making: as more information about the patient is gathered, the uncertainty in downstream variables decreases. In our example, for a young man with a thyroid nodule and suspicious features on ultrasound, the probability of getting a fine-needle aspiration (treatment $A$) is high. In other words, variance in management is low. \Cref{fig:whi_support} illustrates this pattern in real-world data, for which we conduct detailed analyses in \Cref{sec:real-world} and \Cref{sec:whi-appendix}.

Thus, to model the \enquote{informed} nature of clinical decision-making, we focus on the scenarios where $P(T=1 \mid X=x,U=u,R=0)$ in Algorithm~1 is {\em away} from $0.5$. To that end, we define,
\begin{equation} \label{eq:ludb}
    {\cal F} (p) \coloneqq \texttt{Uniform} \big([0.1, p] \cup [1-p, 0.9]\big),\quad p \in (0.1, 0.5].  
\end{equation}
which is a uniform distribution over the union of a {\em low} and a {\em high} range. This ensures that, given the context provided by $X$ and $U$, the uncertainty in downstream variables $T \in \{S,A,Y^0,Y^1\}$ is low upon conditioning on additional patient context. \footnote{Entropy of a Bernoulli random variable $B$, which can be thought of as a measure of its \enquote{uncertainty}, is maximized at $P(B = 1) = 0.5$, and shrinks toward $0$ or $1$.} Next, we define a {\em family} of distributions
\begin{equation} \label{eq:ddef}
    {\cal F} \coloneqq \{ {\cal F} (p) \mid p \in (0.1, 0.5] \}.
\end{equation}
Note that ${\cal F}$ is quite rich and also includes $\texttt{Uniform} ([0.1, 0.9])$ for $p =0.5$ ({\em i.e.}, the whole interval). We fix $0.1$ and $0.9$ as lower and upper bounds to ensure {\em positivity}, which is a central assumption in causal inference. Finally, we state some additional assumptions of our setting, which are also reflected in Figures~\ref{fig:transportability_bias}-\ref{fig:selection_bias_1} (\Cref{fig:selection_bias_2} is studied separately in \Cref{sec:T2SB}). 
\begin{assumption} \label{asm:weak_ex} The following hold for all $r,s,a,u \in \{0,1\}$.
        \textit{Exogeneity of unmeasured covariates} -- $X \indep U \mid R$.
        \textit{Weak transportability} -- $Y^a \indep R \mid X, U$.
        \textit{Weak ignorability} -- $Y^a \indep S \indep A \mid X, U, R$.
        \textit{Positivity} -- $P(R=r,S=s,A=a,U=u \mid X) > 0$.
\end{assumption}
The prefix \enquote{weak} in the transportability and ignorability conditions indicates that they hold only when conditioned on the unmeasured covariate $U$. 
\subsection{Variance in Downstream Variables Drives the Bias} \label{sec:BDVTV}
\paragraph{DAGs miss contextual independencies.}
While DAGs, such as those in \Cref{fig:bias_typs}, are useful in understanding how causal biases emerge, they paint an incomplete picture. Specifically, they do not capture {\em contextual} independencies. To make the point concrete, consider \Cref{fig:confounding_bias} where $U$ influences the treatment assignment $A$ {\em in general}. However, we might have
\begin{equation}
    A \indep U \mid X=x_m,R = 0 \quad \text{whereas} \quad A \not\!\perp\!\!\!\perp U \mid X=x_n,R = 0.\label{eq:coin2}
\end{equation}
That is, $U$ affects $A$ only for {\em some} patients, $X=x_m$, but not others, $X=x_n$ (as in our clinical vignette from the previous section). We discuss a direct result of this observation next.
\paragraph{The bias in the OS varies across patients.} 
Contextual independence of $U$ and $T \in \{S,A,Y^0,Y^1\}$ leads to a key observation: the magnitude of bias in the OS will vary across patient types. To illustrate, consider \Cref{fig:confounding_bias} (confounding bias) and assume that $U$ has an effect on $Y^1$ for all $X=x$. That is, $Y^1 \not\!\perp\!\!\!\perp U \mid X, R=0$. We have, together with \eqref{eq:coin2},
\begin{equation}
    Y^1 \indep A \mid X=x_m,R = 0 \quad \text{whereas} \quad Y^1 \not\!\perp\!\!\!\perp A \mid X=x_n,R = 0  \label{eq:coin4}
\end{equation}
since the \enquote{ignorability of treatment} condition in \Cref{asm:rct_iv} holds for $X = x_m$ but not for $X= x_n$ in the OS ($R = 0$). One then has, $f_1(x_m) = \text{CATE}_{\text{os}} (x_m,A\!=\!1)$ but $f_1(x_n) \neq \text{CATE}_{\text{os}} (x_n,A\!=\!1)$. Consequently, when transportability is not an issue, $b_1(x_m) = 0$ but $b_1(x_n) \neq 0$.

More generally, for patient profiles $X=x$ where the effect of $U$ on $T \in \{S,A,Y^0,Y^1\}$ is stronger, we will observe two key phenomena: first, the magnitude of the bias, $\abs{b(X=x)}$, shall become larger. Second, these high-bias regions also exhibit greater variability in 
$T$, since the influence of $U$ introduces additional unexplained heterogeneity.

For instance, when bias is due to an unmeasured confounder as in \Cref{fig:confounding_bias}, we expect {\em larger variances} in $A$ and $Y$ in {\em severely biased regions}. One can formalize this phenomenon by examining the {\em covariances} of the absolute bias and conditional variance functions: $\Co \big(\abs{b_1(X)}, \V (A | X, S\!=\!1) \big)$ and $\Co \big(\abs{b_1(X)}, \V (Y | X,S\!=\!1,A\!=\!1) \big)$. 

In the next section, we formalize how such covariances between the bias function and conditional variances of various downstream variables can provably distinguish different bias mechanisms.
\subsection{Covariance of the Bias and Conditional Variances} \label{sec:IBMCCV}
We analyze the covariance between the magnitude of bias, $\abs{b_a(X)}$, and the conditional variances of $S$, $A$, and $Y$, under different bias mechanisms in \Cref{fig:bias_typs}. Our main result is that covariance signals construct a \enquote{hash table}, allowing one to differentiate between the bias mechanisms (\Cref{table:bias_table}). 

For simplicity, we focus on the treatment group, $A=1$, and define 
\begin{align}
    p^{T}_{u} &\coloneqq P(T = 1 \mid X, U=u, R=0), \quad T \in \{S,A,Y^1\}. \label{eq:def3} \\ 
    p^{U}_{r} &\coloneqq P(U = 1 \mid R=r). \label{eq:def4}
\end{align}
Below we derive expressions for the bias, $b_1(X)$, under different bias mechanisms in \Cref{fig:bias_typs}.
\begin{restatable}[Quantifying the Bias]{lemma}{biasexpressions}
\label{lemma:biasexpressions}
Under Assumptions \ref{asm:rct_iv} and \ref{asm:weak_ex}, we have the following.
\begin{enumerate}[leftmargin=7mm, noitemsep]
    \item \textbf{Transportability Bias}-- Consider \Cref{fig:transportability_bias} where $S \indep A \indep U \mid X, R=0$. We have,
    \begin{align} \label{eq:trbiaslem}
        b_1(X) = ( p^{U}_{r=1} - p^{U}_{r=0} ) ( p^{Y}_{u=1} - p^{Y}_{u=0} ).
    \end{align}
    \item \textbf{Confounding Bias}-- Consider \Cref{fig:confounding_bias} where $S \indep A, U \mid X, R=0$ and assume that $P(U=1 \mid R=1) = P(U=1 \mid R=0) = 1/2$ in the RCT and OS. We have,
    \begin{align} \label{eq:conbiaslem}
        b_1(X) = ( p^{Y}_{u=1} - p^{Y}_{u=0} )( p^{A}_{u=1} - p^{A}_{u=0} ) / 2 ( p^{A}_{u=1} + p^{A}_{u=0} ).
    \end{align}
    \item \textbf{Selection Bias, Type 1}-- Consider \Cref{fig:selection_bias_1} where $A \indep S, U \mid X, R=0$ and assume that $P(U=1 \mid R=1) = P(U=1 \mid R=0) = 1/2$ in the RCT and OS. We have,
    \begin{align} \label{eq:selbiaslem}
        b_1(X) = ( p^{Y}_{u=1} - p^{Y}_{u=0} ) ( p^{S}_{u=1} - p^{S}_{u=0} ) / 2 ( p^{S}_{u=1} + p^{S}_{u=0} ).
    \end{align}
\end{enumerate}
\end{restatable}
The results are intuitive. In \eqref{eq:trbiaslem}, bias is larger if $U$'s distribution in the RCT and OS differ significantly and $p^{U}_{r=1} - p^{U}_{r=0}$ is big (weak transportability). The contrast between the outcomes for different values of $U$, {\em i.e.} $p^{Y}_{u=1} - p^{Y}_{u=0}$, drives the bias same way. For instance, if $U$ does not have any effect on $Y$ conditioned on $X\!=\!x$, we have $b_1(x)\!=\!0$. The confounding and selection biases in \eqref{eq:conbiaslem} and (\ref{eq:selbiaslem}) react similarly to dependencies between $U$ and $S$, $A$, $Y$.

Building on \Cref{lemma:biasexpressions}, we derive covariance relations between $b_1(X)$ and the conditional variances of downstream variables $A,S,Y$. Let us first introduce a shorthand notation for the covariance signals we are interested in.

\begin{restatable}[Covariance signals]{definition}{CovarianceSignals}
\label{def:covs}
    \begin{align}
        &\bar{\rho} (b_1, S) \coloneqq \Co \big(\abs{b_1(X)},\V(S \mid X, R=0)\big) \nonumber\\
        &\bar{\rho} (b_1, A) \coloneqq \Co \big(\abs{b_1(X)},\V(A \mid X, S=1, R=0)\big) \nonumber\\
        &\bar{\rho} (b_1, Y)  \coloneqq \Co \big(\abs{b_1(X)},\V(Y \mid X, S=1, A=1, R=0)\big)\nonumber
    \end{align}
\end{restatable}

\begin{restatable}[]{theorem}{covarianceresults}
\label{theorem:covarianceresults}
Suppose that Assumptions~\ref{asm:rct_iv} and \ref{asm:weak_ex} hold, and that downstream variables $T \in \{S,A,Y^0,Y^1\}$ are sampled according to Algorithm~1 for some ${\cal F}(p) \in {\cal F}$. For confounding and type-1 selection biases (Figures~\ref{fig:confounding_bias}, \ref{fig:selection_bias_1}), assume that $p_r^U = 1/2$ (see \Cref{eq:def4}) for simplicity. For transportability bias (\Cref{fig:transportability_bias}), let $p_{r=0}^U \neq p_{r=1}^U$ and $p_r^U \sim {\cal F}(p)$. Then, these three bias mechanisms are uniquely characterized by the covariance signals (see \Cref{def:covs}) as in \Cref{table:bias_table}.
\end{restatable}

\subsection{Analyzing Covariances Under Type 2 Selection Bias}
\label{sec:T2SB}
\setlength{\tabcolsep}{2pt}
\begin{wraptable}{r}{0.45\linewidth}
    \vspace{-11pt}
    \centering 
    \footnotesize
    \caption{Covariance between the magnitude of the bias and conditional variances of downstream variables under different bias mechanisms (see Theorems \ref{theorem:covarianceresults} and \ref{theorem:covarianceresults-seltype2}).}
    \begin{tabular}{@{}lccc@{}}
        \toprule
        \toprule
        \textbf{Bias Type} & $\bar{\rho} (b_1, S)$ & $\bar{\rho} (b_1, A)$ & $\bar{\rho} (b_1, Y)$ \\ \midrule 
        \textbf{\begin{tabular}[c]{@{}l@{}}No Bias \end{tabular}} & $= 0 $ & $= 0$ & $= 0$ \\ \midrule
        \textbf{\begin{tabular}[c]{@{}l@{}}Transportability\end{tabular}} &  $= 0$ & $= 0$ & $> 0$  \\ \midrule
        \textbf{\begin{tabular}[c]{@{}l@{}}Confounding\end{tabular}} & $= 0$ & $> 0$ & $> 0$ \\ \midrule 
        \textbf{\begin{tabular}[c]{@{}l@{}}Selection Type 1\end{tabular}} & $> 0$ & $= 0$ & $> 0$  \\ \midrule 
        \textbf{\begin{tabular}[c]{@{}l@{}}Selection Type 2 \end{tabular}} & $\ne 0$ & $\ne 0$ & $\ne 0$  \\ 
        \bottomrule
        \bottomrule
    \end{tabular}
    \label{table:bias_table}
\end{wraptable}
\setlength{\tabcolsep}{6pt}
The bias mechanisms examined previously (Figures \ref{fig:transportability_bias}-\ref{fig:selection_bias_1}) share a fundamental trait: they originate from an {\em unmeasured} variable $U$. This commonality enabled us to employ consistent analytical approaches in deriving both the bias expressions (\Cref{lemma:biasexpressions}) and covariance signals (\Cref{theorem:covarianceresults}). Here, we investigate type 2 selection bias (\Cref{fig:selection_bias_2}), where the root cause is not an unmeasured covariate $U$, but the fact that the selection variable $S$ acts as a collider.

This form of bias, while challenging to detect, is prevalent. Understanding how covariances behave under its various manifestations is therefore crucial. First, given the distinct graph structure (absence of $U$), we present a modified version of \Cref{asm:weak_ex}.
\begin{restatable}[]{assumption}{weakexvtwo}
\label{asmp:weak_ex_v2}
We have {\em transportability} $Y^a \! \indep \! R | X$; and {\em ignorability} $Y^a \! \indep \! A | X, R$ for $a \! \in \! \{0,\!1\}$.
\end{restatable}
Note that due to the $A \rightarrow S$ and $Y \rightarrow S$ edges, ignorability does not hold in the selected cohort: $Y^a \not\!\perp\!\!\!\perp A \mid X, R, S$ which leads to bias as $f_a(X) \ne \text{CATE}_{\text{os}} (x,a)$.

Covariance signals in this case depend on the specific selection mechanism which is characterized by $P(S=1|Y=y, A=a)$. We defer the bulk of the technical analysis to \Cref{sec:type2_appendix} and give the main result directly, which states that the covariance signals will be non-zero in general.

\begin{restatable}[]{theorem}{covarianceresultsseltypetwo}
\label{theorem:covarianceresults-seltype2}
Suppose that \Cref{asmp:weak_ex_v2} holds and downstream variables $T \in \{S,A,Y^0,Y^1\}$ are sampled according to Algorithm~1 for some ${\cal F}(p) \in {\cal F}$. Consider type 2 selection bias in \Cref{fig:selection_bias_2} where $Y^a \not\!\perp\!\!\!\perp A \mid X, S, R=0$. Then, covariance signals (see \Cref{def:covs}) are non-zero in general.
\end{restatable}
For example, consider $P(S = 1 | Y = 1) = 0.9$ and $P(S = 1 | Y=0) = 0.1$, where there is preferential selection of individuals who experience the outcome. We briefly discuss why this yields a positive covariance signal $\bar{\rho} (b_1, Y)$. This selection mechanism systematically overestimates event rates, and crucially, the bias magnitude varies by the risk at baseline. For populations with low event rates, {\em e.g.}, $P(Y^1\!=\!1 | X\!=\!x_1) = 0.1$, the bias is substantial and estimates gravitate toward $0.5$ where the variance is higher. Conversely, in high-risk populations, {\em e.g.,} $P(Y^1\!=\!1 | X\!=\!x_2) = 0.9$, the bias is bounded (at most $0.1$) and the variance in the outcome will be further reduced. 

\subsection{Consistent Estimators of Covariance} \label{sec:AUEC}
In practice, one needs to estimate the covariances, which is not straightforward, especially when $X$ contains continuous variables. We construct an estimator based on instance-wise squared-errors, which is consistent when the nuisance function estimators fitted on the OS are consistent.

\begin{restatable}[Nuisance function estimators in the OS]{definition}{ContinuousEstimator}
    \label{def:nuisance-function-estimators}
    \begin{align}
        \widehat{\eta}_{S} (X) &\coloneqq \widehat{P} (S=1 \mid X, R=0). \label{eq:etas} \\
        \widehat{\eta}_{A} (X) &\coloneqq \widehat{P} (A=1 \mid X, S=1, R=0). \label{eq:etaa} \\
        \widehat{\eta}_{Y} (X) &\coloneqq \widehat{P} (Y=1 \mid X, S=1, A=1, R=0). \label{eq:etay} \\
        \widehat{b}_1 (X) &\coloneqq \widehat{g}_1 (X) - \widehat{f}_1 (X), \label{eq:etab}
    \end{align}
\end{restatable}
where $g_1(X)$ and $f_1(X)$ are defined in \eqref{eq:gdef} and \eqref{eq:fdef}, respectively. $\widehat{g}_1(X)$ is estimated from the RCT data, while the remaining estimators are fitted from the OS data. We use the following estimator, which is for the covariance between the bias function and the {\em squared-error} (SE) of the nuisance function estimators for the target variables $T \in \{S,A,Y\}$.
\begin{equation*}
    \widehat{\bar{\rho}(b_1,T)} = \frac{n}{n-1} \bigg( \frac{1}{n} \sum_{i=1}^n \lvert\widehat{b}_1 (X_i)\rvert (T_i- \widehat{\eta}_{T}(X_i))^2 - \frac{1}{n^2} \sum_{i=1}^n \sum_{j=1}^n \lvert \widehat{b}_1 (X_i)\rvert (T_j - \widehat{\eta}_{T}(X_i))^2 \bigg).
\end{equation*}        

\begin{restatable}[]{theorem}{ThmContinuous}
    \label{thm:continuous}
    Assume that \eqref{eq:etas}-\eqref{eq:etab} are consistent estimators. Then, for all $T\in \{S,A,Y\}$, $\widehat{\bar{\rho}(b_1,T)}$ are consistent estimators of $\bar{\rho}(b_1,T)$. 
\end{restatable}

\section{Synthetic Experiments}
\label{sec:synthetic}
\begin{wrapfigure}{l}{0.5\linewidth}
    \centering
    \vspace{-15pt}
    \includegraphics[width=\linewidth]{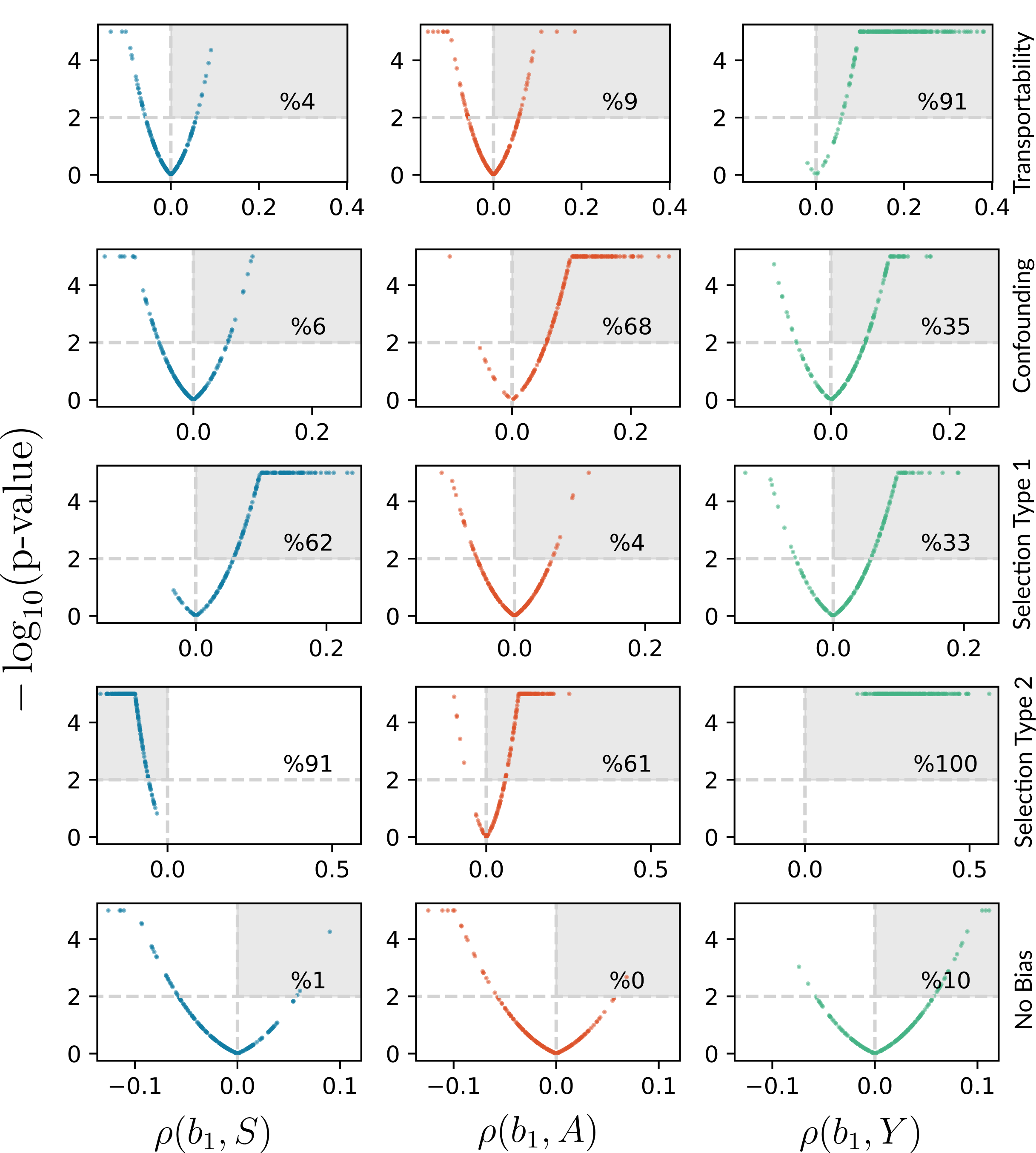}
    \caption{\small Pearson's R (normalized covariance) in synthetic experiments for the bias mechanisms in \Cref{fig:bias_typs}. Percentages denote the ratio of dots in shaded area. $p$-values clipped at $p\!=\!10^{-5}$.}
    \label{fig:synthetic}
    \vspace{-10pt}
\end{wrapfigure}
We run several synthetic experiments under the generative model in \Cref{sec:ja_wnf}. We consider binary covariates, $X \in \{0,1\}^d$, and denote by $X(j)$ the $j$-th covariate. We experiment with $d \in \{5,6,7\}$. We set $P(X(j) = 1|R=1) = 0.4$ in RCT and $P(X(j) = 1|R=0) = 0.6$ in OS to simulate covariate shift between two populations. In the RCT cohort, we use $P (A\!=\!1 | R\!=\!1) \!=\! 0.5$ and $P (S\!=\!1 | R\!=\!1) \!=\! 1$. We set the OS cohort size to $n_{\text{os}} = 50000$ for fitting $\hat{f}_1(X)$, and set aside $n_{\text{val}} = 2000$ OS participants to estimate the correlation signals in \Cref{def:covs} once the nuisance functions in Eq.\ref{eq:etas} - \ref{eq:etab} are fitted. We experiment with different RCT cohort sizes, $n_{\text{rct}} \in \{2000, 50000\}$, for fitting $\hat{g}_1(X)$. 

For each bias mechanism in \Cref{fig:bias_typs}, we run 200 experiments with different seeds. In each run, we sample a $p \sim \texttt{Uniform} [0.2, 0.5]$ and specify ${\cal F} (p)$ in Algorithm~1 accordingly, from which the downstream variables are sampled.

When simulating transportability bias in \Cref{fig:transportability_bias}, we sample $P(U = 1\mid X, R) \sim {\cal F} (p)$ to induce a different distribution of $U$ in the RCT and OS. We set $U$-bias to \texttt{False} for $S,A$ and \texttt{True} for $Y^1$ in Algorithm~1.

In all other experiments, we set $P(U \mid X,R) = 1/2$ to ensure transportability. When simulating confounding bias in \Cref{fig:confounding_bias}, we set $U$-bias to \texttt{False} for $S$ and \texttt{True} for $A,Y^1$ in Algorithm~1. For type 1 selection bias in \Cref{fig:selection_bias_1}, we set $U$-bias to \texttt{False} for $A$ and \texttt{True} for $S,Y^1$. Finally, for type 2 selection bias in \Cref{fig:selection_bias_2}, we set $U$-bias to \texttt{False} for $A$ and $Y^1$, as there is no $U$. We set $P(S=1|Y=1,A=1)=0.9$ and to $0.1$ for other combinations of $Y$ and $A$. In \Cref{sec:type2_appendix}, we report the correlation signals for different parametrizations of type 2 selection bias.

We set $n_{\text{rct}} = 50000$ and $d=6$. The results for $n_{\text{rct}}=2000$ and $d \in \{5,7\}$ are presented in \Cref{app:moresyn}. The (Pearson correlation coefficient, $p$-value) pairs computed from 200 experiment runs are presented in \Cref{fig:synthetic} (each dot is a run). We draw a nominal threshold for the significance of the estimated correlation coefficients at $p =0.01$. On the plots, we write the percentage of experiments where the correlation coefficients come out as suggested in \Cref{table:bias_table} with a significance level $p < 0.01$. The results corroborate our findings: statistically significant correlations between the magnitude of the bias, $\abs{b_1(X)}$, and the squared errors of the nuisance function estimators (see \Cref{def:nuisance-function-estimators}) can characterize the underlying bias mechanism. 

We present additional results for settings where we simulate \textit{combinations} of bias mechanisms in \Cref{app:moresyn}. Further, we report the results for similar experiments where $U$ is modeled to be a continuous covariate in \Cref{app:moresyn}, which yield similar observations, showing that the empirical insights from our methods generalize beyond the settings for which we derive theoretical results.
\section{Real-World Experiments}
\label{sec:real-world} 
Women's Health Initiative (WHI) OS and RCT components yielded conflicting results for the effect of combined hormonal therapy (HRT) on coronary heart disease (CHD) and stroke outcomes, where the OS suggested a protective effect while the RCT revealed increased risk. Follow-up studies found that the OS suffered from ``immortal time'' bias, which can broadly be thought of as a subcategory of type 2 selection bias \citep{prentice2005combined} (see \Cref{sec:sbt2-cc}).
\paragraph{Setup}
Following \citet{prentice2005combined}, we attempt to correct for the immortal time bias by \textit{only} selecting patients into the ``treatment group'' who were not past users of combined HRT but started shortly after enrollment (see \Cref{sec:whi-appendix} for our reproduction of results). We refer to the setting prior to correction as the ``Baseline'' analysis. We refer to the setting after the correction as ``Corrected,'' although it may not be fully unbiased as other types of biases may still be present. Finally, we run another set of experiments where we conceal age and menopausal status in the ``Corrected'' OS to observe how the correlation signals change and refer to this setting as ``Manual bias''.

To compute the correlation signals, we have the following definitions of $S$ and $Y$: $Y\!=\!1$ if the patient experiences the event within the followup period, and $Y\!=\!0$ otherwise. $S\!=\!0$ for patients who are censored or past HRT users who do not meet the inclusion criteria, and $S\!=\!1$ for those included in the analysis.
\paragraph{Results}
In \Cref{fig:main-whi-results} (left), we show covariance signals for all cases. Error bars are 95\% confidence intervals computed over twenty different train/validation splits of the RCT and OS data. 

In the baseline analysis, we see a type 2 selection bias pattern: $\rho (b_1, S) > 0$, $\rho (b_1, A) < 0$, $\rho (b_1, Y) > 0$. Introspecting on $\rho (b_1, A)$, we note that those in the treated group of the OS are largely patients who are more likely to benefit from the treatment, which explains the high bias and low variance in their management, and hence the negative signal.

\begin{figure}[tbp]
\centering
    \includegraphics[width=\linewidth]{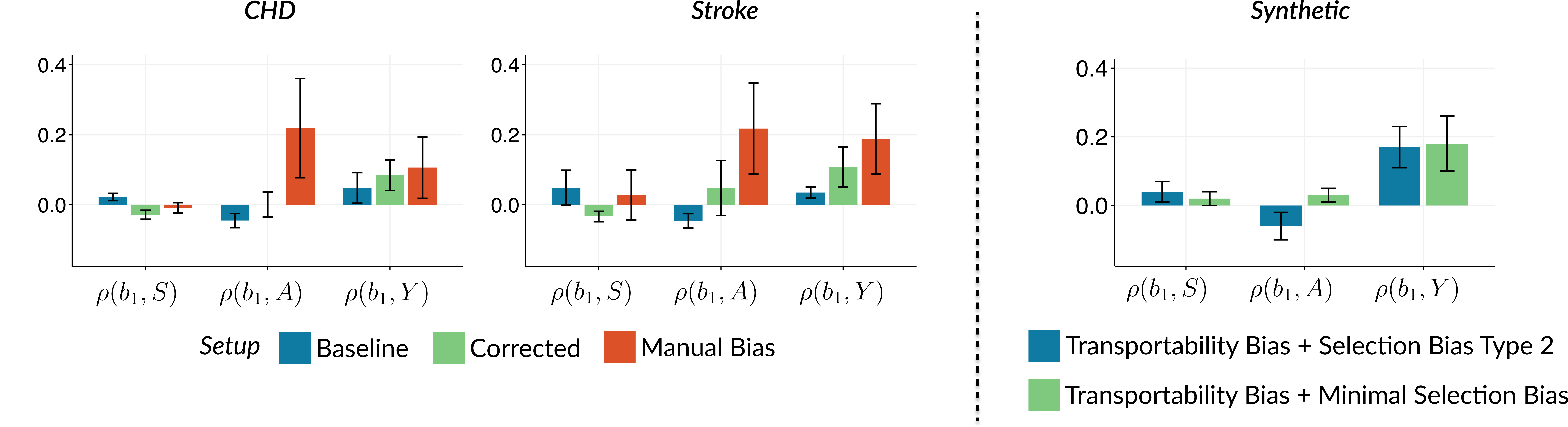}
    \label{fig:stroke+chd+cov}    
    \caption{\small \textbf{Left}: Pearson's R between bias functions and downstream variables for CHD and stroke in WHI experiments. \textbf{Right}: Average correlation signals under different synthetic settings: 1) type 2 selection bias (with selection probabilities matching WHI setting) and transportability bias, 2) minimal selection bias.}
    \label{fig:main-whi-results}
    \vspace{-15pt}
\end{figure}

As we implement the bias correction step, $\rho (b_1, S)$, and $\rho (b_1, A)$ signals tend to expectedly deflate. A curious observation is that $\rho (b_1, Y)$ does not decrease. We hypothesize that there might be two bias mechanisms at play simultaneously, type 2 selection bias and transportability bias. When the former is responsible for most of the bias, one might see an increase (or at least not a decrease) in $\rho (b_1, Y)$ after the correction step, since the residual bias is mostly due to the transportability issue. To test this hypothesis, we conduct a synthetic experiment that incorporates both type 2 selection bias and transportability bias to replicate the WHI setting ($p^S_{00} = 0.9, p^S_{01} = 0.9, p^S_{10} = 0.3, p^S_{11} = 0.1$), shown in~\Cref{fig:main-whi-results} (right). We find that the synthetic combined setting reflects the signals in the real-world ``Baseline'' WHI analysis. Importantly, after correcting only the selection bias, $\rho(b_1,A)$ and $\rho(b_1,S)$ go towards zero, and $\rho(b_1,Y)$ slightly increases, mimicking what we see in the real-world results. For more details on the setup of this experiment, see~\Cref{sec:whi-appendix}.

Finally, we conduct \textit{positive-control} experiments, intentionally concealing age and menopausal status, which are confounders for treatment assignment and outcome (labeled as \enquote{Manual Bias} in \Cref{fig:main-whi-results}). As expected, we note a significant increase in both $\rho (b_1, A) > 0$ and $\rho (b_1, Y) > 0$ signals.
\section{Conclusion}
We uncover a principled connection between the bias function and the predictive performance of nuisance function estimators in observational studies, enabling the identification of underlying bias mechanisms from among common candidates. More broadly, we open a new direction for using model performance diagnostics in observational settings to extract actionable insights about the data-generating process. One limitation is that our analysis relies on a specific generative model. Nonetheless, we support our theoretical assumptions with empirical evidence from real-world data and demonstrate that our analytical insights hold in settings that extend beyond our theoretical framework. Additionally, our current analysis does not address measurement bias—systematic errors in how variables are recorded or observed—which remains an important avenue for future work.

\bibliography{main.bib}
\bibliographystyle{abbrvnat}


\appendix
\newpage
\onecolumn
\section{Appendix}
\label{sec:appendix}

\subsection{Proofs}
Detailed proofs for the theoretical results in the main paper are provided here. All results in the main paper are re-stated in this section for clarity.
\subsubsection{Useful Lemmas}
We start by deriving some auxiliary results that hold true under the general weak \Cref{asm:weak_ex} we make throughout the paper.  
\begin{lemma}
\label{lemma:useful1}
Under \Cref{asm:weak_ex}, we have
\begin{align}
    &f_a(X) \nonumber \\
    &= \sum_{u=0,1} P (Y=1, U=u \mid X, R=0, S=1, A=a) \nonumber \\
    &= \sum_{u=0,1} P (Y^a=1 \mid X, U=u, R=0, S=1, A=a) P (U=u \mid X, R=0, S=1, A=a) \nonumber \\
    &= \sum_{u=0,1} P (Y^a=1 \mid X, U=u, R=0) P (U=u \mid X, R=0, S=1, A=a) \label{eq:l10} \\
    &= \sum_{u=0,1} P (Y^a=1 \mid X, U=u) \frac{P(U=u, S=1,A=a \mid X, R=0) }{P(S=1,A=a \mid X, R=0)}  \label{eq:l11} \\
    &= \sum_{u=0,1} P (Y^a=1 \mid X, U=u) \frac{P(S=1,A=a \mid X, U=u, R=0) P(U=u \mid R=0) }{P(S=1,A=a \mid X, R=0)} \label{eq:l12} \\
    &= \sum_{u=0,1}  P(U=u \mid R=0) P (Y^a=1 \mid X, U=u)\nonumber \\
    &\qquad\qquad\qquad\qquad \times \frac{P(S=1 \mid X, U=u, R=0) P(A=a \mid X, U=u, R=0)}{P(S=1, A=a \mid X, R=0)}. \label{eq:fax_expanded}
\end{align}
\eqref{eq:l10} follows from weak ignorability, \eqref{eq:l11} from weak transportability, \eqref{eq:l12} from exogeneity of $X$ and $U$, and \eqref{eq:fax_expanded} from weak ignorability again.
\end{lemma}
\begin{lemma}
\label{lemma:useful2}
Under \Cref{asm:weak_ex}, we have
\begin{align}
    &P(S=1 \mid X, R=0) \nonumber \\
    &=P(S=1, U=0 \mid X, R=0) + P(S=1, U=1 \mid X, R=0) \nonumber \\
    &=P(S=1 \mid X, U=0, R=0) P(U=0 \mid R=0) \nonumber \\
    &\hspace{20pt}+ P(S=1 \mid X, U=1, R=0) P(U=1 \mid R=0) \label{eq:ule2_2} \\
    &=p^S_{u=0} (1 - p^U_{r=0}) + p^S_{u=1} p^U_{r=0} \nonumber \\
    &=p^S_{u=0} + p^U_{r=0} (p^S_{u=1} - p^S_{u=0}) \label{eq:ule2_res}
\end{align}
where \eqref{eq:ule2_2} follows from exogeneity of $X$ and $U$, and \eqref{eq:ule2_res} from~\eqref{eq:def3},\eqref{eq:def4}.
\end{lemma}
\begin{lemma}
\label{lemma:useful3}
Under \Cref{asm:weak_ex}, we have
\begin{align}
    &P(A=1 \mid X, S=1, R=0) \nonumber \\
    &=P(A=1, U=0 \mid X, S=1, R=0) + P(A=1, U=1 \mid X, S=1, R=0) \nonumber \\
    &=P(A=1 \mid X, U=0, S=1, R=0) P(U=0 \mid X, S=1, R=0) \nonumber \\ 
    &\hspace{20pt} + P(A=1 \mid X, U=1, S=1, R=0) P(U=1 \mid X, S=1, R=0) \nonumber \\
    &=P(A=1 \mid X, U=0, R=0) P(U=0 \mid X, S=1, R=0) \nonumber \\ 
    &\hspace{20pt} + P(A=1 \mid X, U=1, R=0) P(U=1 \mid X, S=1, R=0) \label{eq:ule3_1} \\
    &=p^A_{u=0} P(U=0 \mid X, S=1, R=0) + p^A_{u=1} P(U=1 \mid X, S=1, R=0) \label{eq:ule3_res}
\end{align}
where \eqref{eq:ule3_1} follows from weak ignorability, and \eqref{eq:ule3_res} from \eqref{eq:def3}.
\end{lemma}
\begin{lemma}
\label{lemma:useful4}
Under \Cref{asm:weak_ex}, we have
\begin{align}
    &P(Y=1 \mid X, S=1, A=1, R=0) \nonumber \\
    &=P(Y=1, U=0 \mid X, S=1, A=1, R=0) + P(Y=1, U=1 \mid X, S=1, A=1, R=0) \nonumber \\
    &=P(Y^1=1 \mid X, U=0, S=1, A=1, R=0) P(U=0 \mid X, S=1, A=1, R=0) \nonumber \\
    &\hspace{20pt} + P(Y^1=1 \mid X, U=1, S=1, A=1, R=0) P(U=1 \mid X, S=1, A=1, R=0) \nonumber \\
    &=p^Y_{u=0} P(U=0 \mid X, S=1, A=1, R=0) + p^Y_{u=1} P(U=1 \mid X, S=1, A=1, R=0) \label{eq:ule4_res}
\end{align}
where \eqref{eq:ule4_res} follows from weak ignorability.
\end{lemma}
\subsubsection{Proofs of the Results in \Cref{sec:IBMCCV}}
\biasexpressions*
\begin{proof}[Proof of the Transportability Bias]
    We consider \Cref{fig:transportability_bias} where $S \indep A \indep U \mid X, R=0$. That is, $X$ explains everything in the OS that we are interested in. However, we allow $U \not\!\perp\!\!\!\perp  R$, that is, the unmeasured covariate has different distributions in the RCT and OS. Following \eqref{eq:fax_expanded}, we have 
    \begin{align}
        f_a(X) &= \sum_{u=0,1}  P(U=u \mid R=0) P (Y^a=1 \mid X, U=u) \\ 
        &\qquad\qquad\qquad \times \frac{\cancel{P(S=1 \mid X, U=u, R=0) P(A=a \mid X, U=u, R=0)}}{\cancel{P(S=1 \mid X, R=0) P(A=a \mid X, R=0)} \tag{$S \indep A \indep U \mid X, R=0$}} \nonumber \\
        &= \sum_{u=0,1}  P (Y^a=1 \mid X, U=u) P(U=u \mid R=0). \label{eq:tr_gex}
    \end{align}
    Combining \eqref{eq:cb_gex2} and \eqref{eq:tr_gex} we have
    \begin{align}
        b_a(X) 
        &= g_a(X) - f_a(X) \nonumber \\
        &=P(Y^a = 1 \mid X, U=0) \big( P(U=0 \mid R=1) - P(U=0 \mid R=0) \big) \nonumber \\ 
        &\hspace{10pt} + P(Y^a = 1 \mid X, U=1) \big( P(U=1 \mid R=1) - P(U=1 \mid R=0) \big) \nonumber \\
        &=\big ( P(Y^a = 1 \mid X, U=1) - P(Y^a = 1 \mid X, U=0) \big) \nonumber\\
        &\qquad\qquad \big( P(U=1 \mid R=1) - P(U=1 \mid R=0) \big) \nonumber \\
        &=\big( p^{U}_{r=1} - p^{U}_{r=0} \big) \big( p^{Y}_{u=1} - p^{Y}_{u=0} \big). \label{eq:trs_bex}
    \end{align}
    The result then follows from \eqref{eq:sbt1_bex} and \eqref{eq:def3},\eqref{eq:def4} for $a=1$.
\end{proof}
\begin{proof}[Proof of the Confounding Bias]
    We consider \Cref{fig:confounding_bias} where $S \indep A, U \mid X, R=0$ and $P(U=0 \mid R=0) = P(U=1 \mid R=0) = 1/2$. Following \eqref{eq:fax_expanded}, we have 
    \begin{align}
        &f_a(X) \nonumber \\
        &= \frac{1}{2} \sum_{u=0,1} P (Y^a=1 \mid X, U=u) \frac{P(S=1 \mid X, U=u, R=0) P(A=a \mid X, U=u, R=0)}{P(S=1 \mid X, R=0) P(A=a \mid X, R=0)} \tag{$S \indep A \mid X, R=0$} \\
        &= \frac{1}{2} \sum_{u=0,1} P (Y^a=1 \mid X, U=u) \frac{P(A=a \mid X, U=u, R=0)}{P(A=a \mid X, R=0)} \nonumber \tag{$S \indep U \mid X, R=0$} \\
        &= \frac{1}{2} \sum_{u=0,1} P (Y^a=1 \mid X, U=u) \frac{P(A=a \mid X, U=u, R=0)}{\sum_{u=0,1} P(A=a, U=u \mid X, R=0)} \nonumber \\
        &= \frac{1}{2} \sum_{u=0,1} P (Y^a=1 \mid X, U=u) \frac{P(A=a \mid X, U=u, R=0)}{\sum_{u=0,1} P(A=a \mid X, U=u, R=0) P(U=u \mid X, R=0)} \nonumber \\
        &= \cancel{\frac{1}{2}} \sum_{u=0,1} P (Y^a=1 \mid X, U=u) \frac{P(A=a \mid X, U=u, R=0)}{\sum_{u=0,1} P(A=a \mid X, U=u, R=0) \underbrace{P(U=u \mid R=0)}_{\cancel{1/2}}} \nonumber \\
        &= \frac{\splitdfrac{\Big(P (Y^a=1 \mid X, U=0) P(A=a \mid X, U=0, R=0)}{\qquad\qquad + P (Y^a=1 \mid X, U=1) P(A=a \mid X, U=1, R=0) \Big)}}{P(A=a \mid X, U=0, R=0) + P(A=a \mid X, U=1, R=0)}. \label{eq:cb_fex}
    \end{align}
    Next, note that
    \begin{align}
        g_a(X) 
        &= P(Y^a = 1 \mid X, R=1) \nonumber \\
        &= \sum_{u=0,1} P(Y^a = 1, U=u \mid X, R=1) \nonumber \\ 
        &= \sum_{u=0,1} P(Y^a = 1 \mid X, U=u, R=1) P(U=u \mid X, R=1)   \nonumber \\
        &= \sum_{u=0,1} P(Y^a = 1 \mid X, U=u, R=1) P(U=u \mid R=1)   \label{eq:cb_gex3} \\
        &= \sum_{u=0,1} P(Y^a = 1 \mid X, U=u) P(U=u \mid R=1)   \label{eq:cb_gex2} \\
        &= \frac{P(Y^a = 1 \mid X, U=0) + P(Y^a = 1 \mid X, U=1)}{2}. \label{eq:cb_gex}
    \end{align}
    where \eqref{eq:cb_gex3} and \eqref{eq:cb_gex2} follow from the exogeneity of $X$ and $U$ and weak transportability, respectively, in \Cref{asm:weak_ex}. \eqref{eq:cb_gex} follows from the distribution of $U$ given above.
    
    Combining \eqref{eq:cb_fex} and \eqref{eq:cb_gex}, we have
    \begin{align} 
        b_a(X) &= g_a(X) - f_a(X) \nonumber \\
        &= \frac{\splitdfrac{\big( P (Y^a=1 \mid X, U=0) - P (Y^a=1 \mid X, U=1) \big)} {\qquad\qquad \times \big(P(A=a \mid X, U=1, R=0) - P(A=a \mid X, U=0, R=0) \big)} }{2 \big(P(A=a \mid X, U=1, R=0) + P(A=a \mid X, U=0, R=0) \big)} \label{eq:cb_bex}
    \end{align}
    The result then follows from \eqref{eq:cb_bex} and \eqref{eq:def3} for $a=1$.
\end{proof}
\begin{proof}[Proof of the Selection Bias]
    We consider \Cref{fig:confounding_bias} where $A \indep S, U \mid X, R=0$ and $P(U=0 \mid R=0) = P(U=1 \mid R=0) = 1/2$. Following \eqref{eq:fax_expanded}, we have 
    \begin{align}
        &f_a(X) \nonumber \\
        &= \frac{1}{2} \sum_{u=0,1} P (Y^a=1 \mid X, U=u) \frac{P(S=1 \mid X, U=u, R=0) P(A=a \mid X, U=u, R=0)}{P(S=1 \mid X, R=0) P(A=a \mid X, R=0)} \tag{$A \indep S \mid X, R=0$} \\
        &= \frac{1}{2} \sum_{u=0,1} P (Y^a=1 \mid X, U=u) \frac{P(S=1 \mid X, U=u, R=0)}{P(S=1 \mid X, R=0)} \nonumber \tag{$A \indep U \mid X, R=0$} \label{eq:qs1}.
    \end{align}
    Rest of the steps follow similarly to confounding bias case and we have
    \begin{align} 
        b_a(X) &= g_a(X) - f_a(X) \nonumber \\
        &= \frac{\splitdfrac{\big( P (Y^a=1 \mid X, U=0) - P (Y^a=1 \mid X, U=1) \big)}{\qquad\qquad \times \big(P(S=1 \mid X, U=1, R=0) - P(S=1 \mid X, U=0, R=0) \big)}  }{2 \big(P(S=1 \mid X, U=1, R=0) + P(S=1 \mid X, U=0, R=0) \big)}. \label{eq:sbt1_bex}
    \end{align}
    The result then follows from \eqref{eq:sbt1_bex} and \eqref{eq:def3} for $a=1$.
\end{proof}
\begin{figure}
    \centering
    \includegraphics[width=0.5\linewidth]{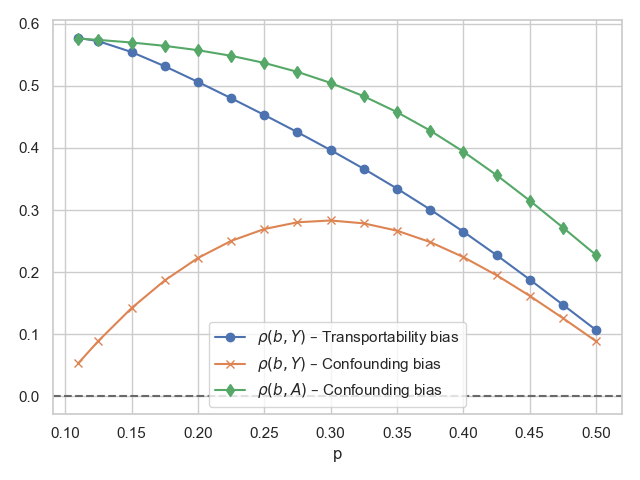}
    \caption{Different values of the correlation between the absolute bias function and the variance of downstream variables for different values of $p \in (0.1, 0.5)$ that parametrize ${\cal F}(p)$ in \eqref{eq:ludb}.}
    \label{fig:covs}
\end{figure}
\covarianceresults*
\begin{proof}[Proof of Covariance Results for the Transportability Bias Case]
    In this case, we have
    \begin{equation}
        b_1(X) = \big( p^{U}_{r=1} - p^{U}_{r=0} \big) \big( p^{Y}_{u=1} - p^{Y}_{u=0} \big), \label{eq:th_trs_1}
    \end{equation}
    by \Cref{eq:trbiaslem} in \Cref{lemma:biasexpressions}. Next,
    \begin{align}
        P(S=1 \mid X, R=0) 
        &= p^S_{u=0} + p^U_{r=0} (p^S_{u=1} - p^S_{u=0}) \tag{\Cref{lemma:useful2}, \eqref{eq:ule2_res}} \\
        &= p^S_{u=0}. \label{eq:th_trs_2}
    \end{align}
    where the last step is due to $p^S_{u=0} = p^S_{u=1}$ since $S \indep U \mid X, R=0$ (see \Cref{fig:transportability_bias}). It follows from \eqref{eq:th_trs_2} that
    \begin{equation}
        \V(S \mid X, R=0) = p^S_{u=0} (1 - p^S_{u=0}). \label{eq:th_trs_3}
    \end{equation}
    Since $p^{U}_{r=1}$, $p^{U}_{r=1}$, $p^{Y}_{u=1}$, $p^{Y}_{u=0}$, and $p^S_{u=0}$ are independent random variables sampled via Algorithm~1, it follows from \eqref{eq:th_trs_1} and \eqref{eq:th_trs_3} that 
    \begin{equation}
        \bar{\rho} (b_1, S) = \Co (\abs{b_1(X)},\V(S | X, R=0)) = 0. \label{eq:th_trs_res_S}
    \end{equation}
    Similarly,
    \begin{align}
        &P(A=1 \mid X,S=1, R=0) \nonumber \\
        &= p^A_{u=0} P(U=0 \mid X, S=1, R=0) + p^A_{u=1} P(U=1 \mid X, S=1, R=0) \tag{\Cref{lemma:useful3}, \eqref{eq:ule3_res}} \\
        &= p^A_{u=0}. \label{eq:th_trs_4}
    \end{align}
    where the last step is due to $p^A_{u=0} = p^A_{u=1}$ since $S \indep A \mid X, R=0$ (see \Cref{fig:transportability_bias}). It follows from \eqref{eq:th_trs_4} that
    \begin{equation}
        \V(A \mid X, S=1, R=0) = p^A_{u=0} (1 - p^A_{u=0}). \label{eq:th_trs_5}
    \end{equation}
    Since $p^{U}_{r=1}$, $p^{U}_{r=1}$, $p^{Y}_{u=1}$, $p^{Y}_{u=0}$, and $p^A_{u=0}$ are independent random variables, it follows from \eqref{eq:th_trs_1} and \eqref{eq:th_trs_5} that 
    \begin{equation}
        \bar{\rho} (b_1, A) = \Co (\abs{b_1(X)},\V(A | X, S=1, R=0)) = 0. \label{eq:th_trs_res_A}
    \end{equation}
    Next,
    \begin{align}
        &P(Y=1 \mid X,S=1,A=1,R=0) \nonumber \\
        &= p^Y_{u=0} P(U=0 \mid X, S=1, A=1, R=0) + p^Y_{u=1} P(U=1 \mid X, S=1, A=1, R=0) \tag{\Cref{lemma:useful4}, \eqref{eq:ule4_res}} \\
        &= p^Y_{u=0} P(U=0 \mid X, R=0) + p^Y_{u=1} P(U=1 \mid X, R=0) \label{eq:th_trs_6} \\
        &= p^Y_{u=0} P(U=0 \mid R=0) + p^Y_{u=1} P(U=1 \mid R=0) \label{eq:th_trs_7} \\
        &= p^Y_{u=0} + p^U_{r=0} (p^Y_{u=1} - p^Y_{u=0}). \label{eq:th_trs_8}
    \end{align}
    where \eqref{eq:th_trs_6} is due to $S \indep A \indep U \mid X, R=0$ and \eqref{eq:th_trs_7} is due to the exogeneity of $X$ and $U$. It follows from \eqref{eq:th_trs_8} that
    \begin{equation} \label{eq:th_trs_9}
        \V(Y \mid X,S=1,A=1,R=0) = (p^Y_{u=0} + p^U_{r=0} (p^Y_{u=1} - p^Y_{u=0}))(1 - p^Y_{u=0} - p^U_{r=0} (p^Y_{u=1} - p^Y_{u=0})).
    \end{equation}
    We can then write
    \begin{align}
        \bar{\rho} (b_1, Y) &= \Co (\abs{b_1(X)}, \V(Y=1 \mid X,S=1,A=1,R=0)) \nonumber \\
        &= \E[\abs{b_1(X)} \V(Y \mid X,S=1,A=1,R=0)] \nonumber\\
        &\qquad\qquad - \E[\abs{b_1(X)}] \E[\V(Y \mid X,S=1,A=1,R=0)] \nonumber \\
        &= \int_{{\cal F}(p)^4} \abs{x-y} \abs{z-w} \big(y + w (x-y)\big) \big(1-y-w(x-y)\big) \mathop{dPx} \mathop{dPy} \mathop{dPz} \mathop{dPw} \nonumber \\
        &\hspace{20pt}-\int_{{\cal F}(p)^4} \abs{x-y} \abs{z-w}  \mathop{dPx} \mathop{dPy} \mathop{dPz} \mathop{dPw} \nonumber\\
        &\hspace{20pt} \times \int_{{\cal F}(p)^3} \big(y + w (x-y)\big) \big(1-y-w(x-y)\big)  \mathop{dPx} \mathop{dPy} \mathop{dPw} \label{eq:th_trs_res_Y}
    \end{align}
    which is nonnegative for all ${\cal F}(p) \in {\cal F}$ (see \Cref{fig:covs} for normalized covariance, {\em i.e.} correlation, $\rho(b_1, Y)$). 
    
    Combining \eqref{eq:th_trs_res_S}, \eqref{eq:th_trs_res_A}, and \eqref{eq:th_trs_res_Y} concludes the proof.
\end{proof}
\begin{proof}[Proof of Covariance Results for the Confounding Bias Case]
    In this case, we have
    \begin{equation}
        b_1(X) = \frac{\big( p^{Y}_{u=1} - p^{Y}_{u=0} \big) \big( p^{A}_{u=1} - p^{A}_{u=0} \big)}{2 \big( p^{A}_{u=1} + p^{A}_{u=0} \big)}, \label{eq:th_con_1}
    \end{equation}
    by \Cref{eq:conbiaslem} in \Cref{lemma:biasexpressions}. Next,
    \begin{align}
        P(S=1 \mid X, R=0) 
        &= p^S_{u=0} + p^U_{r=0} (p^S_{u=1} - p^S_{u=0}) \tag{\Cref{lemma:useful2}, \eqref{eq:ule2_res}} \\
        &= p^S_{u=0}. \label{eq:th_con_2}
    \end{align}
    where the last step is due to $p^S_{u=0} = p^S_{u=1}$ since $S \indep U \mid X, R=0$ (see \Cref{fig:confounding_bias}). It follows from \eqref{eq:th_con_2} that
    \begin{equation}
        \V(S \mid X, R=0) = p^S_{u=0} (1 - p^S_{u=0}). \label{eq:th_con_3}
    \end{equation}
    Since $p^{A}_{u=1}$, $p^{A}_{u=0}$, $p^{Y}_{u=1}$, $p^{Y}_{u=0}$, and $p^S_{u=0}$ are independent random variables sampled via Algorithm~1, it follows from \eqref{eq:th_con_1} and \eqref{eq:th_con_3} that 
    \begin{equation}
        \bar{\rho} (b_1, S) = \Co (\abs{b_1(X)},\V(S | X, R=0)) = 0. \label{eq:th_con_res_S}
    \end{equation}
    Next,
    \begin{align}
        &P(A=1 \mid X,S=1, R=0) \nonumber \\
        &= p^A_{u=0} P(U=0 \mid X, S=1, R=0) + p^A_{u=1} P(U=1 \mid X, S=1, R=0) \tag{\Cref{lemma:useful3}, \eqref{eq:ule3_res}} \\
        &= p^A_{u=0} P(U=0 \mid R=0) + p^A_{u=1} P(U=1 \mid R=0) \label{eq:th_con_4} \\
        &= \frac{p^A_{u=0} + p^A_{u=1}}{2}  \label{eq:th_con_5}
    \end{align}
    where \eqref{eq:th_con_4} follows from $S  \indep U \mid X, R=0$ (see \Cref{fig:confounding_bias}) and the exogeneity of $X$ and $U$, and \eqref{eq:th_con_5} from $P(U=0 \mid R=0) = 1/2$. It follows from \eqref{eq:th_con_5} that
    \begin{equation}
        \V(A \mid X, S=1, R=0) = \frac{p^A_{u=0} + p^A_{u=1}}{2} \frac{2 - p^A_{u=0} - p^A_{u=1}}{2}. \label{eq:th_con_6}
    \end{equation}
    We start with the following key observation: both the bias in \eqref{eq:th_con_1} and the conditional variance of $A$ in \eqref{eq:th_con_6} are functions of $p^A_{u=1}$ and $p^A_{u=0}$. Hence, their covariance $\bar{\rho} (b_1, A)$ will be nonzero in general. We give an intuitive explanation before the formal result.

    Note that $\abs{b_1(X)} \propto \abs{p^A_{u=1} - p^A_{u=0}}$, and that $\V(A \mid X, S=1, R=0)$ is effectively the variance of a Bernoulli random variable with $p = \frac{p^A_{u=1} + p^A_{u=0}}{2}$, which increases toward $p=0.5$ and decreases toward $p=0$ and $p=1$. Since probabilities are confined in the $[0,1]$ range, $\abs{b_1(X)}$, is maximized when $p^A_{u=0} = 0$ and $p^A_{u=1} = 1$, or vice versa. In those cases, $\frac{p^A_{u=1} + p^A_{u=0}}{2} = 0.5$, that is,  $\V(A \mid X, S=1, R=0)$ is also maximized. 
    
    In short, $\abs{b_1(X)}$ and $\V(A \mid X, S=1, R=0)$ tend to align. This is intuitive: when $U$'s effect on $A$ is significant, the bias is expected to be larger. Also, not accounting for $U$ induces more uncertainty on the estimated probability of treatment assignment, which is reflected in its variance. 

    Formally, we are interested in
    \begin{align}
        \bar{\rho} (b_1, A) &= \Co (\abs{b_1(X)}, \V(A \mid X, S=1, R=0))) \nonumber \\
        &= \E[\abs{b_1(X)} \V(A \mid X, S=1, R=0)] - \E[\abs{b_1(X)}] \E[\V(A \mid X, S=1, R=0)] \nonumber \\
        &= \int_{{\cal F}(p)^4} \frac{\abs{x-y} \abs{z-w}}{2(z+w)} \frac{z+w}{2} \frac{2-z-w}{2}  \mathop{dPx} \mathop{dPy} \mathop{dPz} \mathop{dPw} \nonumber \\
        &\hspace{20pt} - \int_{{\cal F}(p)^4} \frac{\abs{x-y} \abs{z-w}}{2(z+w)}  \mathop{dPx} \mathop{dPy} \mathop{dPz} \mathop{dPw} \times \int_{{\cal F}(p)^2} \frac{z+w}{2} \frac{(2-z-w)}{2}  \mathop{dPz} \mathop{dPw} \label{eq:th_con_res_A}
    \end{align}
    which is nonnegative for all ${\cal F}(p) \in {\cal F}$ (see \Cref{fig:covs} for normalized covariance, {\em i.e.} correlation, $\rho(b_1, A)$). 
    
    Next,
    \begin{align}
        &P(Y=1 \mid X,S=1,A=1,R=0) \nonumber \\
        &= p^Y_{u=0} P(U=0 \mid X, S=1, A=1, R=0) + p^Y_{u=1} P(U=1 \mid X, S=1, A=1, R=0) \tag{\Cref{lemma:useful4}, \eqref{eq:ule4_res}} 
    \end{align}
    Here, observe that
    \begin{align}
        &P(U \mid X, S=1, A=1, R=0) \nonumber \\
        &= \frac{P(U, S=1, A=1 \mid X, R=0)}{P(S=1,A=1 \mid X,R=0)} \nonumber \\
        &= \frac{P(U, A=1 \mid X, R=0) P(S=1 \mid X, R=0)}{P(S=1 \mid X,R=0) P(A=1 \mid X,R=0)} \tag{$S \indep A,U \mid X,R=0$} \\
        &= P(U \mid X, A=1, R=0) \nonumber \\
        &= \frac{P(A=1 \mid X, U, R=0) P(U \mid X,R=0)}{P(A=1 \mid X,R=0)} \nonumber \\
        &= \frac{P(A=1 \mid X, U, R=0) P(U \mid R=0)}{P(A=1 \mid X, U=0, R=0) P(U=0 \mid R=0) + P(A=1 \mid X, U, R=0) P(U=1 \mid R=0)} \tag{$X \indep U \mid R=0$} \\
        &= \frac{p^A_U}{p^A_{u=0} + p^A_{u=1}} \label{eq:th_con_7}
    \end{align}
    where \eqref{eq:th_con_7} is due to $P(U=0 \mid R=0)=1/2$. Plugging \eqref{eq:th_con_7} back in we have
    \begin{align} \label{eq:th_con_8}
        P(Y=1 \mid X,S=1,A=1,R=0) = \frac{p^Y_{u=0} p^A_{u=0} + p^Y_{u=1} p^A_{u=1}}{p^A_{u=0} + p^A_{u=1}}
    \end{align}
    It follows from \eqref{eq:th_con_8} that
    \begin{equation} \label{eq:th_con_9}
        \V(Y \mid X,S=1,A=1,R=0) = \frac{p^Y_{u=0} p^A_{u=0} + p^Y_{u=1} p^A_{u=1}}{p^A_{u=0} + p^A_{u=1}} \Big( 1 - \frac{p^Y_{u=0} p^A_{u=0} + p^Y_{u=1} p^A_{u=1}}{p^A_{u=0} + p^A_{u=1}} \Big). 
    \end{equation}
    We can then write
    \begin{align}
        \bar{\rho} (b_1, Y) &= \Co (\abs{b_1(X)}, \V(Y=1 \mid X,S=1,A=1,R=0)) \nonumber \\
        &= \E[\abs{b_1(X)} \V(Y \mid X,S=1,A=1,R=0)] \nonumber\\ 
        &\qquad\qquad - \E[\abs{b_1(X)}] \E[\V(Y \mid X,S=1,A=1,R=0)] \nonumber \\
        &= \int_{{\cal F}(p)^4} \frac{\abs{x-y} \abs{z-w}}{2(z+w)} \frac{xz+yw}{z+w} \big( 1- \frac{xz+yw}{z+w} \big)  \mathop{dPx} \mathop{dPy} \mathop{dPz} \mathop{dPw} \nonumber \\
        &\hspace{20pt} - \int_{{\cal F}(p)^4} \frac{\abs{x-y} \abs{z-w}}{2(z+w)}  \mathop{dPx} \mathop{dPy} \mathop{dPz} \mathop{dPw} \\
        &\hspace{20pt} \times \int_{{\cal F}(p)^4} \frac{xz+yw}{z+w} \big( 1- \frac{xz+yw}{z+w} \big)  \mathop{dPx} \mathop{dPy} \mathop{dPz} \mathop{dPw} \label{eq:th_con_res_Y}
    \end{align}
    which is nonnegative for all ${\cal F}(p) \in {\cal F}$ (see \Cref{fig:covs} for normalized covariance, {\em i.e.} correlation, $\rho(b_1, Y)$). 
    
    Combining \eqref{eq:th_con_res_S}, \eqref{eq:th_con_res_A}, and \eqref{eq:th_con_res_Y} concludes the proof.
\end{proof}
\begin{proof}[Proof of Covariance Results for the Selection Bias Case]
    In this case, we have
    \begin{equation}
        b_1(X) = \frac{\big( p^{Y}_{u=1} - p^{Y}_{u=0} \big) \big( p^{S}_{u=1} - p^{S}_{u=0} \big)}{2 \big( p^{S}_{u=1} + p^{S}_{u=0} \big)}, \label{eq:th_sel_1}
    \end{equation}
    by \Cref{eq:selbiaslem} in \Cref{lemma:biasexpressions}. Next,
    \begin{align}
        P(S=1 \mid X, R=0) 
        &= p^S_{u=0} + p^U_{r=0} (p^S_{u=1} - p^S_{u=0}) \tag{\Cref{lemma:useful2}, \eqref{eq:ule2_res}} \\
        &= \frac{p^S_{u=0} + p^S_{u=1}}{2}. \label{eq:th_sel_2}
    \end{align}
    where the last step is due to from $p^U_{r=0} = 1/2$. It follows from \eqref{eq:th_sel_2} that
    \begin{equation}
        \V(S=1 \mid X, R=0) = \frac{p^S_{u=0} + p^S_{u=1}}{2} \frac{(2 - p^S_{u=0} - p^S_{u=1})}{2}. \label{eq:th_sel_3}
    \end{equation}
    Notice that \eqref{eq:th_sel_1} and \eqref{eq:th_sel_3} follow the same formats as \eqref{eq:th_con_1} and \eqref{eq:th_con_6}, respectively. Then, following the same steps in the derivation of \eqref{eq:th_con_res_A}, we have
    \begin{equation} \label{eq:th_sel_res_S}
        \bar{\rho} (b_1, S) > 0.
    \end{equation}
    for all ${\cal F}(p) \in {\cal F}$.

    Next,
    \begin{align}
        &P(A=1 \mid X,S=1, R=0) \nonumber \\
        &= p^A_{u=0} P(U=0 \mid X, S=1, R=0) + p^A_{u=1} P(U=1 \mid X, S=1, R=0) \tag{\Cref{lemma:useful3}, \eqref{eq:ule3_res}} \\
        &= p^A_{u=0} P(U=0 \mid X, S=1, R=0) + p^A_{u=0} P(U=1 \mid X, S=1, R=0) \label{eq:th_sel_5} \\
        &= p^A_{u=0}  \label{eq:th_sel_6}
    \end{align}
    where \eqref{eq:th_sel_5} follows from $A  \indep U \mid X, R=0$, hence $p^A_{u=0} = p^A_{u=1}$  (see \Cref{fig:selection_bias_1}). \eqref{eq:th_sel_6} follows simply because $P(U=0 \mid X, S=1, R=0) + P(U=1 \mid X, S=1, R=0) = 1$. It follows from \eqref{eq:th_sel_6} that
    \begin{equation}
        \V(A \mid X, S=1, R=0) = p^A_{u=0} (1 - p^A_{u=0}). \label{eq:th_sel_7}
    \end{equation}
    Since $p^{S}_{u=1}$, $p^{S}_{u=0}$, $p^{Y}_{u=1}$, $p^{Y}_{u=0}$, and $p^A_{u=0}$ are independent random variables sampled via Algorithm~1, it follows from \eqref{eq:th_sel_1} and \eqref{eq:th_sel_7} that 
    \begin{equation}
        \bar{\rho} (b_1, A) = \Co (\abs{b_1(X)},\V(S | X, R=0)) = 0. \label{eq:th_sel_res_A}
    \end{equation}
    Next,
    \begin{align}
        &P(Y=1 \mid X,S=1,A=1,R=0) \nonumber \\
        &= p^Y_{u=0} P(U=0 \mid X, S=1, A=1, R=0) + p^Y_{u=1} P(U=1 \mid X, S=1, A=1, R=0) \tag{\Cref{lemma:useful4}, \eqref{eq:ule4_res}} 
    \end{align}
    Here, observe that
    \begin{align}
        &P(U \mid X, S=1, A=1, R=0) \nonumber \\
        &= \frac{P(U, S=1, A=1 \mid X, R=0)}{P(S=1,A=1 \mid X,R=0)} \nonumber \\
        &= \frac{P(U, S=1 \mid X, R=0) P(A=1 \mid X, R=0)}{P(S=1 \mid X,R=0) P(A=1 \mid X,R=0)} \tag{$A \indep S,U \mid X,R=0$} \\
        &= P(U \mid X, S=1, R=0) \nonumber \\
        &= \frac{P(S=1 \mid X, U, R=0) P(U \mid X,R=0)}{P(S=1 \mid X,R=0)} \nonumber \\
        &= \frac{P(S=1 \mid X, U, R=0) P(U \mid R=0)}{P(S=1 \mid X, U=0, R=0) P(U=0 \mid R=0) + P(S=1 \mid X, U, R=0) P(U=1 \mid R=0)} \tag{$X \indep U \mid R=0$} \\
        &= \frac{p^S_U}{p^S_{u=0} + p^S_{u=1}} \label{eq:th_sel_8}
    \end{align}
    where \eqref{eq:th_sel_8} is due to $P(U=0 \mid R=0)=1/2$. Plugging \eqref{eq:th_sel_8} back in we have
    \begin{align} \label{eq:th_sel_9}
        P(Y=1 \mid X,S=1,A=1,R=0) = \frac{p^Y_{u=0} p^S_{u=0} + p^Y_{u=1} p^S_{u=1}}{p^S_{u=0} + p^S_{u=1}}
    \end{align}
    It follows from \eqref{eq:th_sel_9} that
    \begin{equation} \label{eq:th_sel_10}
        \V(Y \mid X,S=1,A=1,R=0) = \frac{p^Y_{u=0} p^S_{u=0} + p^Y_{u=1} p^S_{u=1}}{p^S_{u=0} + p^S_{u=1}} \Big( 1 - \frac{p^Y_{u=0} p^S_{u=0} + p^Y_{u=1} p^S_{u=1}}{p^S_{u=0} + p^S_{u=1}} \Big). 
    \end{equation}
    Notice that \eqref{eq:th_sel_1} and \eqref{eq:th_sel_10} follow the same formats as \eqref{eq:th_con_1} and \eqref{eq:th_con_9}, respectively. Then, following the same steps in the derivation of \eqref{eq:th_con_res_Y}, we have
    \begin{equation} \label{eq:th_sel_res_Y}
        \bar{\rho} (b_1, Y) > 0.
    \end{equation}
    for all ${\cal F}(p) \in {\cal F}$.

    Combining \eqref{eq:th_sel_res_S}, \eqref{eq:th_sel_res_A}, and \eqref{eq:th_sel_res_Y} concludes the proof.
\end{proof}

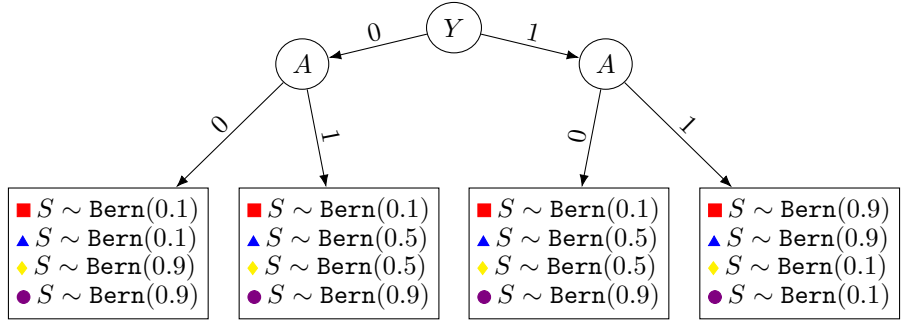
\begin{figure}[htbp]
    \centering
    \begin{tikzpicture}
        \node[state] (y) [] {$Y$};
        \node[state] (a1) [below left =of y, xshift=-0.5cm, yshift=1cm] {$A$}; 
        \node[state] (a2) [below right =of y, xshift=0.5cm, yshift=1cm] {$A$}; 
        \node[state, rectangle, align=left] (l1) [below     left =of a1,xshift=0.01cm, yshift=-0.4cm] {\marksymbol{square*}{red} $S \sim \texttt{Bern}(0.1)$ \\ \marksymbol{triangle*}{blue} $S \sim \texttt{Bern}(0.1)$ \\ \marksymbol{diamond*}{yellow} $S \sim \texttt{Bern}(0.9)$ \\ \marksymbol{oplus*}{blue!50!red} $S \sim \texttt{Bern}(0.9)$};
        \node[state, rectangle, align=left] (l2) [below right =of a1,xshift=-2.1cm, yshift=-0.4cm] {\marksymbol{square*}{red} $S \sim \texttt{Bern}(0.1)$ \\\marksymbol{triangle*}{blue} $S \sim \texttt{Bern}(0.5)$ \\ \marksymbol{diamond*}{yellow} $S \sim \texttt{Bern}(0.5)$ \\ \marksymbol{oplus*}{blue!50!red} $S \sim \texttt{Bern}(0.9)$};
        \node[state, rectangle, align=left] (r1) [below left =of a2,xshift=2.1cm, yshift=-0.4cm] {\marksymbol{square*}{red} $S \sim \texttt{Bern}(0.1)$ \\ \marksymbol{triangle*}{blue} $S \sim \texttt{Bern}(0.5)$ \\ \marksymbol{diamond*}{yellow} $S \sim \texttt{Bern}(0.5)$ \\ \marksymbol{oplus*}{blue!50!red} $S \sim \texttt{Bern}(0.9)$};
        \node[state, rectangle, align=left] (r2) [below right =of a2,xshift=-0.01cm, yshift=-0.4cm] {\marksymbol{square*}{red} $S \sim \texttt{Bern}(0.9)$ \\ \marksymbol{triangle*}{blue} $S \sim \texttt{Bern}(0.9)$ \\ \marksymbol{diamond*}{yellow} $S \sim \texttt{Bern}(0.1)$ \\ \marksymbol{oplus*}{blue!50!red} $S \sim \texttt{Bern}(0.1)$};        
        \path (y) edge node[el] {$0$} (a1);
        \path (y) edge node[el] {$1$} (a2);
        \path (a1) edge node[el] {$0$} (l1);
        \path (a1) edge node[el] {$1$} (l2);
        \path (a2) edge node[el] {$0$} (r1);
        \path (a2) edge node[el] {$1$} (r2);
    \end{tikzpicture}
    \caption{Four different specifications of the selection model. See \eqref{eq:t2sb_c1} - (\ref{eq:t2sb_c4}) for the values of the correlation signals in each case.}
    \label{fig:sel_type2_DGP}
\end{figure}

\subsubsection{Proof of the Results in \Cref{sec:T2SB}} \label{sec:type2_appendix}
As alluded to in the main paper, type 2 selection bias requires modified assumptions given its distinct graph structure, i.e. absence of $U$. We repeat \Cref{asmp:weak_ex_v2} here for clarity.
\weakexvtwo*

The covariance signals in this case depend on the specific selection mechanism characterized by $P(S=1|Y=y,A=a)$. We consider the specifications shown in \Cref{fig:sel_type2_DGP} of the selection model in our synthetic experiments as well as in our theoretical analysis. 
Conceptually, in our analysis, we show that the sign of the covariance signals differ based on the selection mechanism, implying that the covariances are nonzero, in general. First, we give adjusted definitions that we will use in our derivations. 
\begin{definition}(Abbreviations for treated group in OS, selection bias type 2)
\label{def:abbs-type2}
\begin{align}
    p^A &:= P(A = 1|X,R=0)\\
    p^{S}_{{Y=y,A=a}} &:= P(S = 1|X,A=a,Y = y,R=0) \label{eq:psy}\\
    p^{Y}_{a} &:= P(Y^a = 1 | X) \label{eq:py}
\end{align}
\end{definition}
We need not condition on $A$ and $R$ in (\ref{eq:py}) due to our assumptions of ignorability and transportability, respectively. Furthermore, we are unable to eliminate the conditioning on $A$ and $Y$ in (\ref{eq:psy}) due to the lack of conditional ignorability in the type 2 selection bias setting.

We can express the bias mechanism, $b_1(X)$, under type 2 selection bias as follows, 
\begin{restatable}[Selection Bias Type 2]{lemma}{seltypetwoexpression}
\label{lemma:biasexpression-seltype2}
Consider \Cref{fig:selection_bias_2} where $Y^a \not\!\perp\!\!\!\perp A \mid X, S, R=0$. Then, we have, 
\begin{align} \label{eq:selbias2lem}
b_1(X) = p^{Y}_1\left(1 - \frac{p^{S}_{Y=1,A=1}}{p^{S}_{Y=1,A=1}p^{Y}_1 + p^{S}_{Y=0,A=1}(1-p^{Y}_1)}\right)
\end{align}
\end{restatable}

\begin{proof}
As a shorthand, we will define $p^{S}_{Y=y,A=a} = p^{S}_{ya}$. We begin by showing expressions for $g_a(X)$ and $f_a(X)$ individually. We have, 
\begin{align*}
    g_a(X) &= P(Y^a =1 \mid X, R = 1) \\
    &= \sum_{s\in\{0,1\}} P(Y^a = 1, S=s \mid X, R = 1) \\
    &= \sum_{s\in\{0,1\}} P(Y^a=1 \mid X, S=s, R = 1) P(S=s \mid X, R=1)  \\
    &= \sum_{s\in\{0,1\}} P(Y^a=1 \mid X, R = 1) P(S=s \mid X, R=1), \quad\quad \text{given } Y^a\indep S \mid  X,R\\
    &= \sum_{s\in\{0,1\}} P(Y^a=1 \mid X) P(S=s \mid X, R=1), \quad\quad \text{given } Y^a \indep R \mid X \\
    &= P(Y^a=1 \mid X) \sum_{s\in\{0,1\}} P(S=s \mid X, R=1) \\
    &= P(Y^a=1 \mid X)
\end{align*}
Next, we have, 
\begin{align*}
    f_a(X) &= P(Y=1 \mid X, S=1,A=a,R=0) \\ 
    &= \frac{P(Y=1, S=1 \mid X, A=a,R=0)}{P(S=1 \mid X, A=a,R=0)} \\
    &= \frac{P(S=1 \mid X, Y = 1, A=a,R=0)P(Y =1 \mid X, A=a, R=0)}{P(S=1 \mid X, A=a,R=0)} \\
    &= \frac{P(S=1 \mid X, Y = 1, A=a,R=0)P(Y^a =1 \mid X, A=a, R=0)}{P(S=1 \mid X, A=a,R=0)} \\
    &= \frac{P(S=1 \mid X, Y = 1, A=a,R=0)P(Y^a =1 \mid X, R=0)}{P(S=1 \mid X, A=a,R=0)},\quad\quad \text{given } Y^a \indep A \mid X, R\\
    &= \frac{P(S=1 \mid X, Y = 1, A=a,R=0)P(Y^a =1 \mid X)}{P(S=1 \mid X, A=a,R=0)},\quad\quad \text{given } Y^a \indep R \mid X\\
\end{align*}
Combining the above two expressions, we have,
\begin{align*}
    b_a(X) &= g_a(X) - f_a(X)\\
    &= P(Y^a =1 \mid X) - \frac{P(S=1 \mid X, Y = 1, A=a,R=0)P(Y^a =1 \mid X)}{P(S=1 \mid X, A=a,R=0)} \\ 
    &= P(Y^a =1 \mid X) \left(1 - \frac{P(S=1 \mid X, Y = 1, A=a,R=0)}{P(S=1 \mid X, A=a,R=0)}\right) \\
    &= P(Y^a =1 \mid X)\nonumber\\
    &\qquad \left(1 - \frac{P(S=1 \mid X, Y = 1, A=a,R=0)}{\sum_{y\in\{0,1\}}P(S=1 \mid X, Y = y, A=a,R=0)P(Y = y \mid X, A=a,R=0) }\right) \\
    &= P(Y^a =1 \mid X) \nonumber\\
    &\qquad \left(1 - \frac{P(S=1 \mid X, Y = 1, A=a,R=0)}{\sum_{y\in\{0,1\}}P(S=1 \mid X, Y = y, A=a,R=0)P(Y^a = y \mid X, A=a,R=0) }\right) \\
    &= P(Y^a =1 \mid X) \left(1 - \frac{P(S=1 \mid X, Y = 1, A=a,R=0)}{\sum_{y\in\{0,1\}}P(S=1 \mid X, Y = y, A=a,R=0)P(Y^a = y \mid X) }\right) \\
    &= p^{Y}_a\left(1 - \frac{p^{S}_{1a}}{p^{S}_{1a}p^{Y}_a + p^{S}_{0a}(1-p^{Y}_a)} \right)
\end{align*}
completing the proof.
\end{proof}

\covarianceresultsseltypetwo* 

\begin{proof}
From \Cref{lemma:biasexpression-seltype2}, we have, 
\begin{align*}
    b_1(X) = p^{Y}_1\left(1 - \frac{p^{S}_{11}}{p^{S}_{11}p^{Y}_1 + p^{S}_{01}(1-p^{Y}_1)} \right)
\end{align*}

We will compute the variance of each of the target variables, $Y,A,$ and $S$, followed by the covariance of the bias and the conditional variances. Starting with $Y$, we have,
\begin{align}
    P(Y =1 &\mid X, S=1, A=1, R=0) \nonumber\\
    &= \frac{P(Y=1, S=1 \mid X, A=1, R=0)}{P(S=1 \mid X, A=1, R=0)}\nonumber\\ 
    &= \frac{P(S=1 \mid X, Y=1, A=1, R=0)P(Y = 1 \mid X, A=1, R=0)}{P(S=1 \mid X, A=1, R=0)} \nonumber\\
    &= \frac{P(S=1 \mid X, Y=1, A=1, R=0)P(Y^1 = 1 \mid X, A=1, R=0)}{P(S=1 \mid X, A=1, R=0)} \nonumber\\
    &= \frac{P(S=1 \mid X, Y=1, A=1, R=0)P(Y^1 = 1 \mid X)}{P(S=1 \mid X, A=1, R=0)},\quad\text{by Assmp.}~\ref{asmp:weak_ex_v2} \nonumber\\
    &= \frac{P(S=1 \mid X, Y=1, A=1, R=0)P(Y^1 = 1 \mid X)}{\sum_y P(S=1 \mid X,Y = y,A=1, R=0)P(Y = y \mid X,A=1,R=0)} \nonumber\\
    &= \frac{P(S=1 \mid X, Y=1, A=1, R=0)P(Y^1 = 1 \mid X)}{\sum_y P(S=1 \mid X,Y = y,A=1, R=0)P(Y^1 = y \mid X,A=1,R=0)} \nonumber\\
    &= \frac{P(S=1 \mid X, Y=1, A=1, R=0)P(Y^1 = 1 \mid X)}{\sum_y P(S=1 \mid X,Y = y,A=1, R=0)P(Y^1 = y \mid X)},\quad\text{by Assmp.}~\ref{asmp:weak_ex_v2} \nonumber\\
    &= \frac{p^{S}_{11}p^{Y}_1}{p^{S}_{11}p^{Y}_1 + p^{S}_{01}(1-p^{Y}_1)} \label{eq:py_sel2_hello}
\end{align}

It follows from~\eqref{eq:py_sel2_hello} that,
\begin{align}
    \V(Y \mid X, S=1,A=1,R=0) = \left( \frac{p^{S}_{11}p^{Y}_1}{p^{S}_{11}p^{Y}_1 + p^{S}_{01}(1-p^{Y}_1)} \right) \left( 1 - \frac{p^{S}_{11}p^{Y}_1}{p^{S}_{11}p^{Y}_1 + p^{S}_{01}(1-p^{Y}_1)} \right)
\end{align}
Next, we consider $S$. We have,
\begin{align}
    P(S &= 1 | X, R = 0)\nonumber\\ 
    &= \sum_{y,a}P(S = 1 \mid X, Y = y, A = a, R = 0)P(Y = y, A = a | X, R=0) \nonumber\\
    &= \sum_{y,a}P(S = 1 \mid X, Y = y, A = a, R=0) \nonumber\\
    &\qquad\qquad P(Y = y \mid X, A = a, R= 0)P(A = a | X, R = 0) \nonumber\\
    &= \sum_{y,a}P(S = 1 \mid X, Y = y, A = a, R=0) \nonumber\\
    &\qquad\qquad P(Y^a = y \mid X, A = a, R= 0)P(A = a | X, R = 0) \nonumber\\
    &= \sum_{y,a}P(S = 1 \mid X, Y^a = y, A = a, R=0) \nonumber\\ &\qquad\qquad P(Y^a = y \mid X)P(A = a | X, R = 0) \quad\text{by Assmp.}~\ref{asmp:weak_ex_v2} \nonumber\\
    &= \underbrace{p^{S}_{11}p^{Y}_1 p^A + p^{S}_{10}p^{Y}_0(1-p^A) + p^{S}_{01}(1-p^{Y}_1) p^A + p^{S}_{00}(1-p^{Y}_0)(1-p^A)}_{p^{S\mid X}} \label{eq:sel2biaspsx}
\end{align}
It follows from \eqref{eq:sel2biaspsx} that,
\begin{align}
    \V(S \mid X, R=0) = p^{S \mid X} (1 - p^{S \mid X})
\end{align}
Next, we consider $A$. We have,
\begin{align}
    &P(S=1,A=1 \mid X, R = 0) \nonumber\\
    &= P(S=1 \mid X, A=1, R = 0)P(A = 1 | X, R=0) \nonumber\\
    &= \sum_y P(S=1 \mid X, A=1, Y = y, R = 0)P(Y^1=y \mid X, A=1, R =0 )P(A = 1 | X, R=0) \nonumber\\
    &= \sum_y P(S=1 \mid X, A=1, Y = y, R = 0)P(Y^1=y \mid X)P(A = 1 | X, R=0)\quad\text{by Assmp.}~\ref{asmp:weak_ex_v2} \nonumber\\
    &= p^{S}_{11}p^{Y}_1 p^A + p^{S}_{01}(1-p^{Y}_1) p^A \label{eq:psaxr}
\end{align}
Hence, we have from \eqref{eq:sel2biaspsx} and \eqref{eq:psaxr}, 
\begin{align}
    P(A = 1 &\mid X,S=1,R=0) \nonumber\\
    &= \frac{P(S=1,A=1 \mid X, R = 0)}{P(S = 1 \mid X, R=0)} \nonumber \\
    &= \underbrace{\frac{p^{S_{11}}p^{Y_1} p^A + p^{S_{01}}(1-p^{Y_1}) p^A}{p^{S_{11}}p^{Y_1} p^A + p^{S_{10}}p^{Y_0}(1-p^A) + p^{S_{01}}(1-p^{Y_1}) p^A + p^{S_{00}}(1-p^{Y_0})(1-p^A)}}_{p^{A\mid X,S}} \label{eq:sel2biaspba}
\end{align}
It follows from \eqref{eq:sel2biaspba} that,
\begin{align}
    \V(A \mid X, S=1, R=0) = p^{A \mid X, S} (1 - p^{A \mid X, S})
\end{align}
We compute each covariance signal as in the prior proofs for $T \in \{A,Y,S\}$ and for each of the selection mechanisms specified in \Cref{fig:sel_type2_DGP}, 
\begin{align}
    \bar{\rho} (b_1, T) &= \Co (\abs{b_1(X)}, \V(T \mid \cdot)) \nonumber \\
    &= \E[\abs{b_1(X)} \V(T \mid \cdot))] - \E[\abs{b_1(X)}] \E[\V(A \mid \cdot))] \nonumber 
\end{align}
followed by normalization by the standard deviations of the bias and variance. We get the following results for each selection mechanism,
\begin{align}
    &\{\marksymbol{square*}{red} P^S_{00} = 0.1, P^S_{01} = 0.1, P^S_{10} = 0.1, P^S_{11} = 0.9\}\nonumber\\
    &\hspace{80pt}~\text{---}~\rho(b_1,S) = -0.66, \rho(b_1, A) = 0.013, \rho(b_1,Y) = 0.98 \label{eq:t2sb_c1} \\
    &\{\marksymbol{triangle*}{blue} P^S_{00} = 0.1, P^S_{01} = 0.5, P^S_{10} = 0.5, P^S_{11} = 0.9\}\nonumber\\
    &\hspace{80pt}~\text{---}~\rho(b_1,S) = 0.33, \rho(b_1, A) = -0.010, \rho(b_1,Y) = 0.95 \label{eq:t2sb_c2} \\
    &\{\marksymbol{diamond*}{yellow} P^S_{00} = 0.9, P^S_{01} = 0.5, P^S_{10} = 0.5, P^S_{11} = 0.1\}\nonumber\\
    &\hspace{80pt}~\text{---}~\rho(b_1,S) = -0.37, \rho(b_1, A) = 0.058, \rho(b_1,Y) = 0.98 \label{eq:t2sb_c3} \\
    &\{\marksymbol{oplus*}{blue!50!red} P^S_{00} = 0.9, P^S_{01} = 0.9, P^S_{10} = 0.9, P^S_{11} = 0.1\}\nonumber\\
    &\hspace{80pt}~\text{---}~\rho(b_1,S) = 0.63, \rho(b_1, A) = 0.052, \rho(b_1,Y) = 0.97 \label{eq:t2sb_c4}
\end{align}

The above implies that the covariance signals are nonzero, in general, completing the proof.

\end{proof}

\subsubsection{Proofs of the Results in \Cref{sec:AUEC}}
We repeat our shorthand notation of the covariance signals as well as the definitions of the nuisance function estimators for clarity. 
\CovarianceSignals*
\ContinuousEstimator*

\begin{lemma}
\label{lem:squared-error}
Assume that \eqref{eq:etas}-\eqref{eq:etab} are consistent estimators. Then, for any $T \in \{S,A,Y\}$, we have in the limit as $n \rightarrow \infty$, 
\begin{equation}
    \E_{X,T}\brc{\lvert\widehat{b}_1 (X)\rvert  (T_j - \widehat{\eta}_{T}(X))^2} = \E_{X}\brc{\lvert b_1 (X)\rvert\V_{T|X}(T|X,\cdot)},
\end{equation}
where $T \mid X,\cdot \in \{Y \mid X, S = 1, A = 1, R =0; A \mid X, S = 1, R =0; S \mid X, R = 0 \}$, for $Y$, $A$, and $S$, respectively.
\end{lemma}
\begin{proof}
    For simplicity, we set $T = Y$, since the analysis for $T = \{A,S\}$ is the same.
    Observe, 
    \begin{align*}
        \E_{X,Y}\brc{\lvert\widehat{b}_1 (X)\rvert  (Y - \widehat{\eta}_{Y}(X))^2} &=  \E_{X}[\lvert\widehat{b}_1 (X)\rvert E_{Y|X}[(Y-\widehat{\eta}_Y(X))^2]]\\ 
        &= \E_{X}[\lvert b_1 (X)\rvert  E_{Y|X}[(Y-\eta_Y(X))^2]],\quad\text{by consistency assumption}\\ 
        &= \E_{X}[\lvert b_1 (X)\rvert \V_{Y|X}(Y|X,S=1,A=1,R=0)]. 
    \end{align*}
The last step follows from the definition of variance and the fact that the true function $\eta_Y(X)$ is the conditional mean of the outcomes, $Y$. By the same argument, it follows that $\E_{X,Y}\brc{(Y - \widehat{\eta}_{Y}(X))^2} = \E_{X}[\V_{Y|X}(Y|X,S=1,A=1,R=0)]$. 
\end{proof}

Consider the estimator of the covariance between the magnitude of the bias and the conditional variance of the target variables, 
\begin{align*}
    &\widehat{\bar{\rho}(b_1,T)} = \frac{n}{n-1} \bigg( \frac{1}{n} \sum_{i=1}^n \lvert\widehat{b}_1 (X_i)\rvert (T_i- \widehat{\eta}_{T}(X_i))^2 - \frac{1}{n^2} \sum_{i=1}^n \sum_{j=1}^n \lvert \widehat{b}_1 (X_i)\rvert (T_j - \widehat{\eta}_{T}(X_i))^2 \bigg). 
\end{align*}

\ThmContinuous*
\begin{proof}
    For simplicity, we will proceed for $T = Y$ and $\widehat{\eta}_Y(X)$, i.e. the estimator of the outcome function. The analysis is the same for $S$ and $\widehat{\eta}_S(X)$ as well as $A$ and $\widehat{\eta}_A(X)$. We consider the terms of the expectation below separately. We will first show the {\em asymptotic unbiasedness} of our estimator. Note, 
    \begin{align*}
        &\E_{X,Y}\brc{\widehat{\bar{\rho}(b_1,Y)}}\\
        &= \frac{n}{n-1} \Biggl( \underbrace{\E_{X,Y} \brc{\frac{1}{n} \sum_{i=1}^{n} \lvert\widehat{b}_1 (X_i)\rvert \cdot (Y_i - \widehat{\eta}_Y(X_i))^2}}_{(1)} - \\
        &\qquad\qquad\qquad \underbrace{\E_{X,Y} \brc{\frac{1}{n} \sum_{i=1}^{n} \lvert\widehat{b}_1 (X_i)\rvert \cdot \frac{1}{n} \sum_{i=1}^n (Y_i - \widehat{\eta}_Y(X_i))^2}}_{(2)} \Biggr)
    \end{align*}
    Consider term (1), where we have, 
    \begin{align*}
        \E_{X,Y} \brc{\frac{1}{n} \sum_{i=1}^{n} \lvert\widehat{b}_1 (X_i)\rvert \cdot (Y_i - \widehat{\eta}_Y(X_i))^2} &= \frac{1}{n} \sum_{i=1}^n \E_{X,Y} \brc{\lvert\widehat{b}_1 (X_i)\rvert \cdot (Y_i - \widehat{\eta}_Y(X_i))^2} \\ 
        &= \frac{1}{n} \sum_{i=1}^n \lvert\widehat{b}_1 (X_i)\rvert \cdot (Y_i - \widehat{\eta}_Y(X_i))^2 \\
        &= \E_{X,Y} \brc{ \lvert\widehat{b}_1 (X)\rvert \cdot (Y - \widehat{\eta}_Y(X))^2 }\\
        &= \E_{X} \brc{ \lvert b_1 (X)\rvert \cdot \V(Y|X,S=1,A=1,R=0)}, \\ 
        &\qquad\qquad \text{by Lemma~\ref{lem:squared-error}}
    \end{align*}

    Subsequently, we will write $\V(Y|X,S=1,A=1,R=0)$ as $\V(Y|X,\cdot)$. For term (2), we have, 
    \begin{align}
        &\E_{X,Y} \brc{\frac{1}{n} \sum_{i=1}^{n} \lvert\widehat{b}_1 (X_i)\rvert \cdot \frac{1}{n} \sum_{i=1}^n (Y_i - \widehat{\eta}_Y(X_i))^2}\nonumber\\
        &= \E_{X}\brc{\frac{1}{n} \sum_{i=1}^{n} \lvert\widehat{b}_1 (X_i)\rvert} \E_{X,Y} \brc{\frac{1}{n} \sum_{i=1}^n (Y_i - \widehat{\eta}_Y(X_i))^2} + \Co\para{\frac{1}{n} \sum_{i=1}^{n} \lvert\widehat{b}_1 (X_i)\rvert, \frac{1}{n} \sum_{i=1}^n (Y_i - \widehat{\eta}_Y(X_i))^2} \nonumber\\ 
        &= \E_X\brc{\lvert\widehat{b}_1(X)\rvert} \E_X\brc{\V(Y|X,\cdot)} + \Co\para{\frac{1}{n} \sum_{i=1}^{n} \lvert\widehat{b}_1(X_i)\rvert, \frac{1}{n} \sum_{i=1}^n (Y_i - \widehat{\eta}_Y(X_i))^2}, \nonumber\\
        &\hspace{300pt}\text{ by Lemma~\ref{lem:squared-error}} \nonumber\\ 
        &= \E_{X}[\lvert\widehat{b}_1(X)\rvert]\E_{X}[\V(Y|X,\cdot)] - \E\brc{\frac{1}{n} \sum_{i=1}^{n} \lvert\widehat{b}_1(X_i)\rvert} \E\brc{\frac{1}{n} \sum_{i=1}^n (Y_i - \widehat{\eta}_Y(X_i))^2} \nonumber\\
        &\qquad\qquad + \E\brc{\frac{1}{n} \sum_{i=1}^{n} \lvert\widehat{b}_1(X_i)\rvert \cdot \frac{1}{n} \sum_{i=1}^n (Y_i - \widehat{\eta}_Y(X_i))^2}\nonumber\\ 
        &= \E_{X}[\lvert\widehat{b}_1(X)\rvert]\E_{X}[\V(Y|X,\cdot)] - \E_X\brc{\lvert\widehat{b}_1(X)\rvert} \E_{X,Y}\brc{(Y - \widehat{\eta}_Y(X))^2}\nonumber\\
        &\qquad\qquad + \E\brc{\frac{1}{n} \sum_{i=1}^{n} \lvert\widehat{b}_1(X_i)\rvert \cdot \frac{1}{n} \sum_{i=1}^n (Y_i - \widehat{\eta}_Y(X_i))^2}\nonumber\\ 
        &= \E_{X}[\lvert\widehat{b}_1(X)\rvert]\E_{X}[\V(Y|X,\cdot)] - \E_{X}[\lvert\widehat{b}_1(X)\rvert]\E_{X}[\V(Y|X,\cdot)]\nonumber\\
        &\qquad\qquad + \E\brc{\frac{1}{n} \sum_{i=1}^{n} \lvert\widehat{b}_1(X_i)\rvert \cdot \frac{1}{n} \sum_{i=1}^n (Y_i - \widehat{\eta}_Y(X_i))^2}, \text{ by Lemma~\ref{lem:squared-error}}\nonumber\\ 
        &= \E\brc{\frac{1}{n^2} \sum_{i=1}^{n} \sum_{j=1}^{n} \lvert\widehat{b}_1(X_i)\rvert (Y_j - \widehat{\eta}_Y(X_j))^2 }\nonumber\\ 
        &= \frac{1}{n^2} \sum_{i=1}^{n} \sum_{j=1}^{n} \lvert\widehat{b}_1(X_i)\rvert  (Y_j - \widehat{\eta}_Y(X_j))^2 \nonumber\\ 
        &= \frac{1}{n^2} \para{n \E[\lvert\widehat{b}_1(X)\rvert (Y - \widehat{\eta}_Y(X))^2] + (n^2 - n) \E[\lvert\widehat{b}_1(X)\rvert] \E[(Y - \widehat{\eta}_Y(X))^2]} \label{eq:secondtolast}\\
        &= \frac{1}{n^2} \para{n \E[\lvert b_1(X)\rvert \V(Y|X,\cdot)] + (n^2 - n) \E[\lvert b_1(X)\rvert] \E[\V(Y|X,\cdot)]},\nonumber\\
        &\qquad\qquad \text{ by Lemma~\ref{lem:squared-error} and consistency of $\widehat{b}_1$}\nonumber
    \end{align}

    \eqref{eq:secondtolast} is derived from the fact that when $i \ne j$, $\lvert b_1(X_i)\rvert$ and $Y_j - \widehat{\eta}_Y(X_j)$ are independent, which occurs in $n^2 - n$ of the $n^2$ terms. The remaining $n$ terms are dependent, since $i=j$ and they rely on the same $X_i$. We combine the terms to get the desired result, 
    \begin{align*}
        &\E_{X,Y}\brc{\widehat{\bar{\rho}(b_1,Y)}} \\
        &= \frac{n}{n-1} \para{\E\brc{\lvert b_1 (X)\rvert \V(Y|X,\cdot)} - \frac{1}{n} \E[\lvert b_1 (X)\rvert \V(Y|X,\cdot)] -  \frac{n-1}{n}\E[\lvert b_1 (X)\rvert]\E[\V(Y|X,\cdot)]}\\ 
        &= \E\brc{\lvert b_1 (X)\rvert \V(Y|X,\cdot)} - \E[\lvert b_1 (X)\rvert]\E[\V(Y|X,\cdot)] \\ 
        &= \Co(\lvert b_1 (X)\rvert, \V(Y|X,\cdot))\\
        &= \bar{\rho}(b_1,Y).
    \end{align*}

    Next, note that our estimator’s variance converges to zero since it is a linear combination of sample means of i.i.d. bounded random variables. Asymptotic unbiasedness, which we show above, together with the vanishing variance imply consistency. The result follows in a few lines using Chebyshev's inequality and the triangle rule as follows.
    
    Let $\hat{\theta}_n$ be a sequence of estimators for $\theta$ such that:

\[
\lim_{n \to \infty} \mathbb{E}[\hat{\theta}_n] = \theta \quad \text{(asymptotically unbiased)}
\]

\[
\lim_{n \to \infty} \text{Var}(\hat{\theta}_n) = 0 \quad \text{(vanishing variance)}
\]

For any $\varepsilon > 0$, by Chebyshev's inequality:

\[
P(|\hat{\theta}_n - \mathbb{E}[\hat{\theta}_n]| \geq \varepsilon) \leq \frac{\text{Var}(\hat{\theta}_n)}{\varepsilon^2}
\]

By triangle inequality:

\[
|\hat{\theta}_n - \theta| \leq |\hat{\theta}_n - \mathbb{E}[\hat{\theta}_n]| + |\mathbb{E}[\hat{\theta}_n] - \theta|
\]

For sufficiently large $n$, by asymptotic unbiasedness, $|\mathbb{E}[\hat{\theta}_n] - \theta| < \frac{\varepsilon}{2}$. Therefore:

\[
P(|\hat{\theta}_n - \theta| \geq \varepsilon) \leq P\left(|\hat{\theta}_n - \mathbb{E}[\hat{\theta}_n]| \geq \frac{\varepsilon}{2}\right) \leq \frac{4 \cdot \text{Var}(\hat{\theta}_n)}{\varepsilon^2}
\]

Taking the limit as $n \to \infty$:

\[
\lim_{n \to \infty} P(|\hat{\theta}_n - \theta| \geq \varepsilon) \leq \lim_{n \to \infty} \frac{4 \cdot \text{Var}(\hat{\theta}_n)}{\varepsilon^2} = 0
\]

which proves $\hat{\theta}_n$ is consistent and we are done.
\end{proof}
\subsection{Additional DAGs for Selection Bias Type 2}
\label{sec:add-dags-t2sb}
See \Cref{fig:bias_types_sb2} for additional DAGs reflecting type 2 selection bias. 

\begin{figure}[t]
    \centering
    \begin{subfigure}{0.3\textwidth}
        \centering
        \begin{tikzpicture}
            \node[state] (x) [] {$X$};
            \node[state] (a) [below left =of x,xshift=0.6cm, yshift=0.3cm] {$A$};
            \node[state] (y) [right =of a] {$Y$};
            \node[state] (s) [right =of y, xshift=0cm] {$S$};
            
            \path (x) edge  (a);
            \path (x) edge  (y);
            \path (x) edge  (s);
            \path (a) edge  (y);

             \path (y) edge[red]  (s);
             \path[red] (a) edge[bend right=55] (s);
        \end{tikzpicture}
        \label{fig:selection_bias_2_base}
        \caption{Selection Bias – Type 2 \\ (Base)}
    \end{subfigure}%
    \hfill
    \begin{subfigure}{0.3\textwidth}
        \centering
        \begin{tikzpicture}
            \node[state] (x) [] {$X$};
            \node[state] (a) [below left =of x,xshift=0.6cm, yshift=0.3cm] {$A$};
            \node[state] (y) [right =of a] {$Y$};
            \node[state] (s) [right =of y, xshift=0cm] {$S$};
            \node[state,red] (e) [below right =of a, yshift=.5cm] {$E$}; 
            
            \path (x) edge  (a);
            \path (x) edge  (y);
            \path (x) edge  (s);
            \path (a) edge  (y);

             \path (y) edge[red]  (s);
             \path[red] (a) edge[bend right=25] (e);
             \path[red] (e) edge[bend right=25] (s);
             
        \end{tikzpicture}
        \label{fig:selection_bias_2_mediator}
        \caption{Selection Bias – Type 2 \\ (Mediator)}
    \end{subfigure}
    \hfill
    \begin{subfigure}{0.3\textwidth}
        \centering
        \begin{tikzpicture}
            \node[state] (x) [] {$X$};
            \node[state] (a) [below left =of x,xshift=0.6cm, yshift=0.3cm] {$A$};
            \node[state] (y) [right =of a] {$Y$};
            \node[state] (s1) [right =of y, xshift=-0.5cm] {$S_1$};
            \node[state,red] (s2) [below =of s1, yshift=0.5cm] {$S_2$};
            
            \path (x) edge  (a);
            \path (x) edge  (y);
            \path (x) edge  (s1);
            \path (a) edge  (y);

             \path (y) edge[red]  (s1);
             \path[red] (a) edge[bend right=25] (s1);
             \path[red] (s1) edge[red] (s2);
        \end{tikzpicture}
        \label{fig:selection_bias_3_downstream}
        \caption{Selection Bias – Type 2 (Downstream Effect)}
    \end{subfigure}%
    \caption{Common equivalent variants of Selection Bias Type 2.}
    \label{fig:bias_types_sb2}
\end{figure}
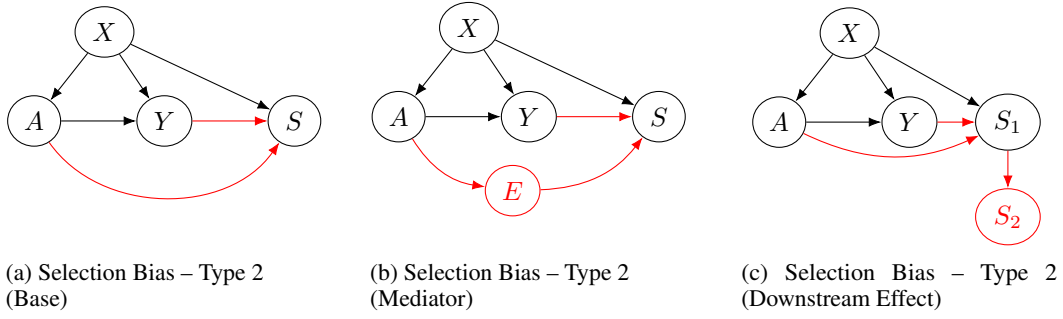

\section{Additional Experimental Results} \label{app:asr}
We ran all experiments on a standard 12-core CPU machine. 

\subsection{Synthetic Results} \label{app:moresyn}
\begin{figure}[ht]
    \centering
    \begin{subfigure}[b]{0.44\linewidth}
        \centering
        \includegraphics[width=\linewidth]{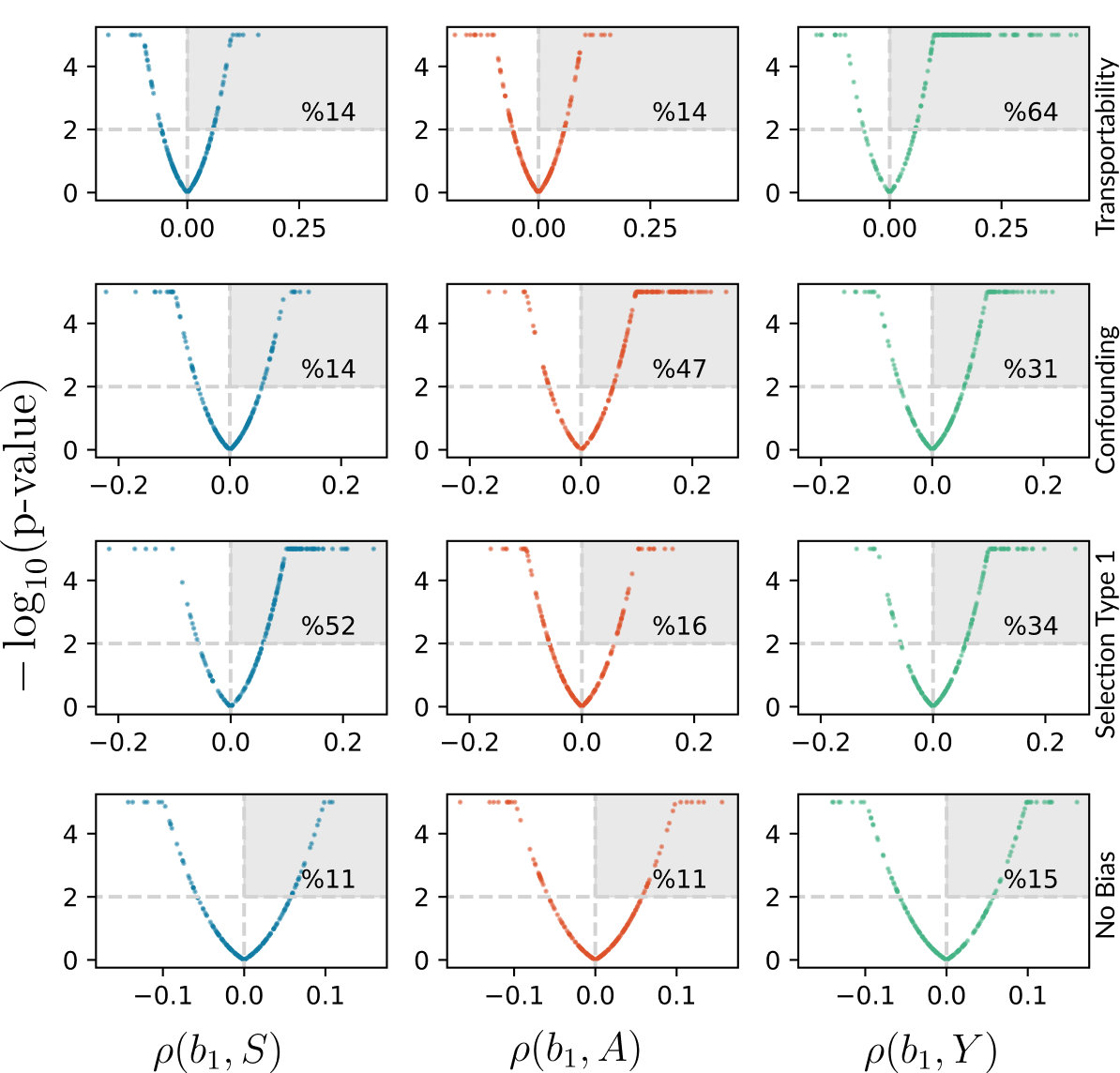}
        \caption{$n_{\text{rct}} = 2000, d=5$}
        \label{fig:2kappd5}
    \end{subfigure} \hspace{20pt}
    \begin{subfigure}[b]{0.44\linewidth}
        \centering
        \includegraphics[width=\linewidth]{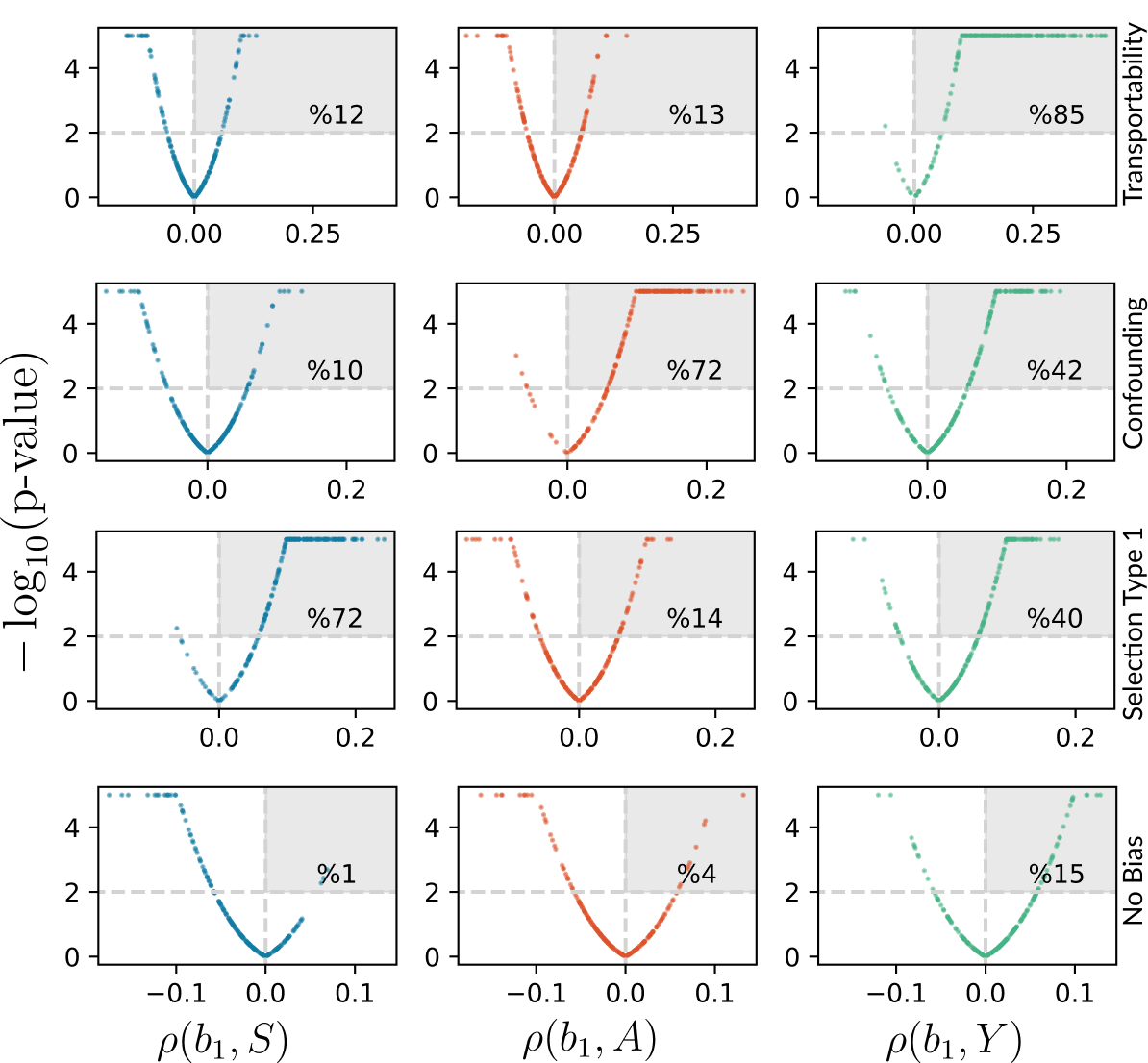}
        \caption{$n_{\text{rct}} = 50000, d=5$}
        \label{fig:50kappd5}
    \end{subfigure}
    \begin{subfigure}[b]{0.44\linewidth}
        \centering
        \includegraphics[width=\linewidth]{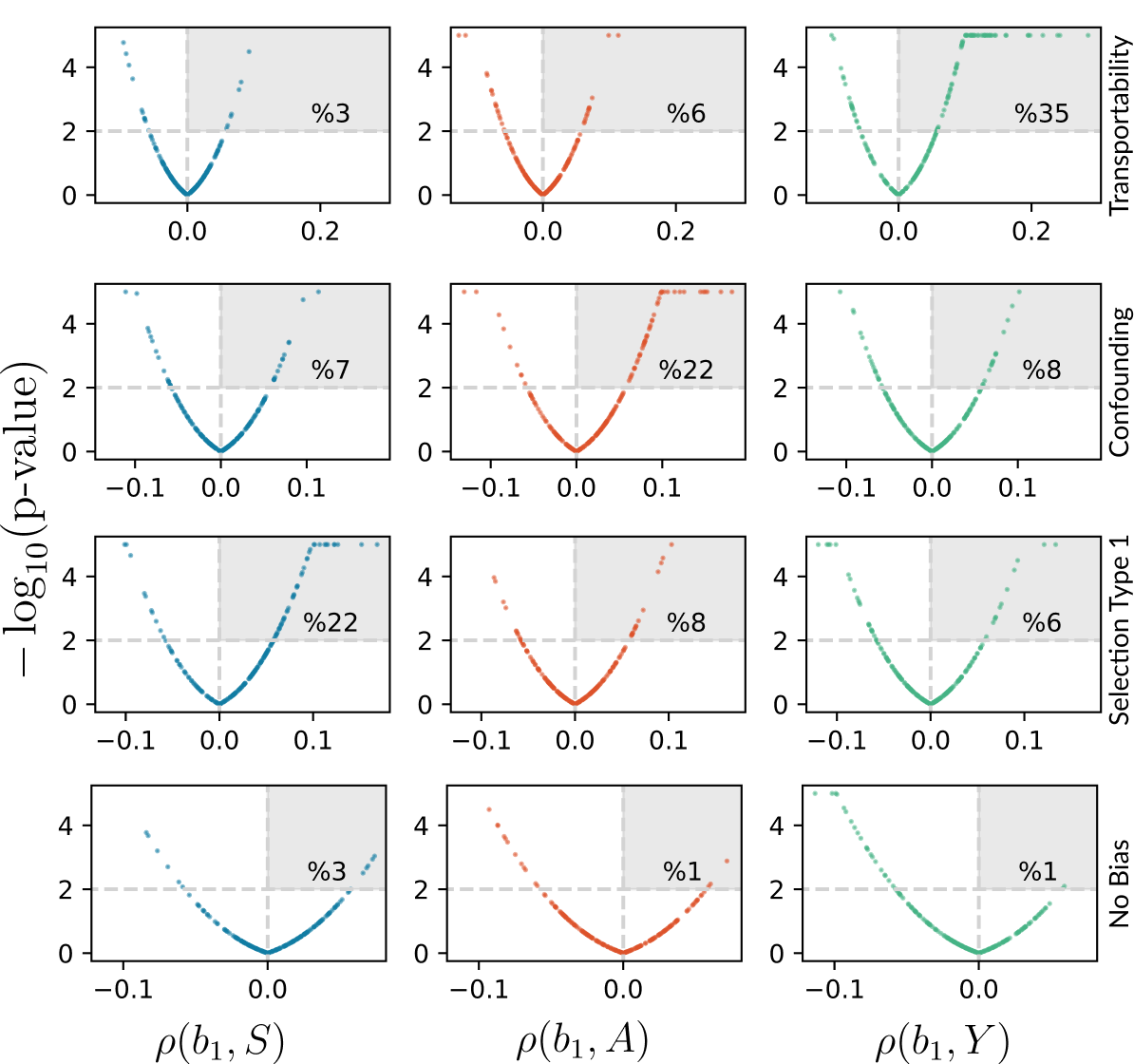}
        \caption{$n_{\text{rct}} = 2000, d=7$}
        \label{fig:2kappd7}
    \end{subfigure}\hspace{20pt}
    \begin{subfigure}[b]{0.44\linewidth}
        \centering
        \includegraphics[width=\linewidth]{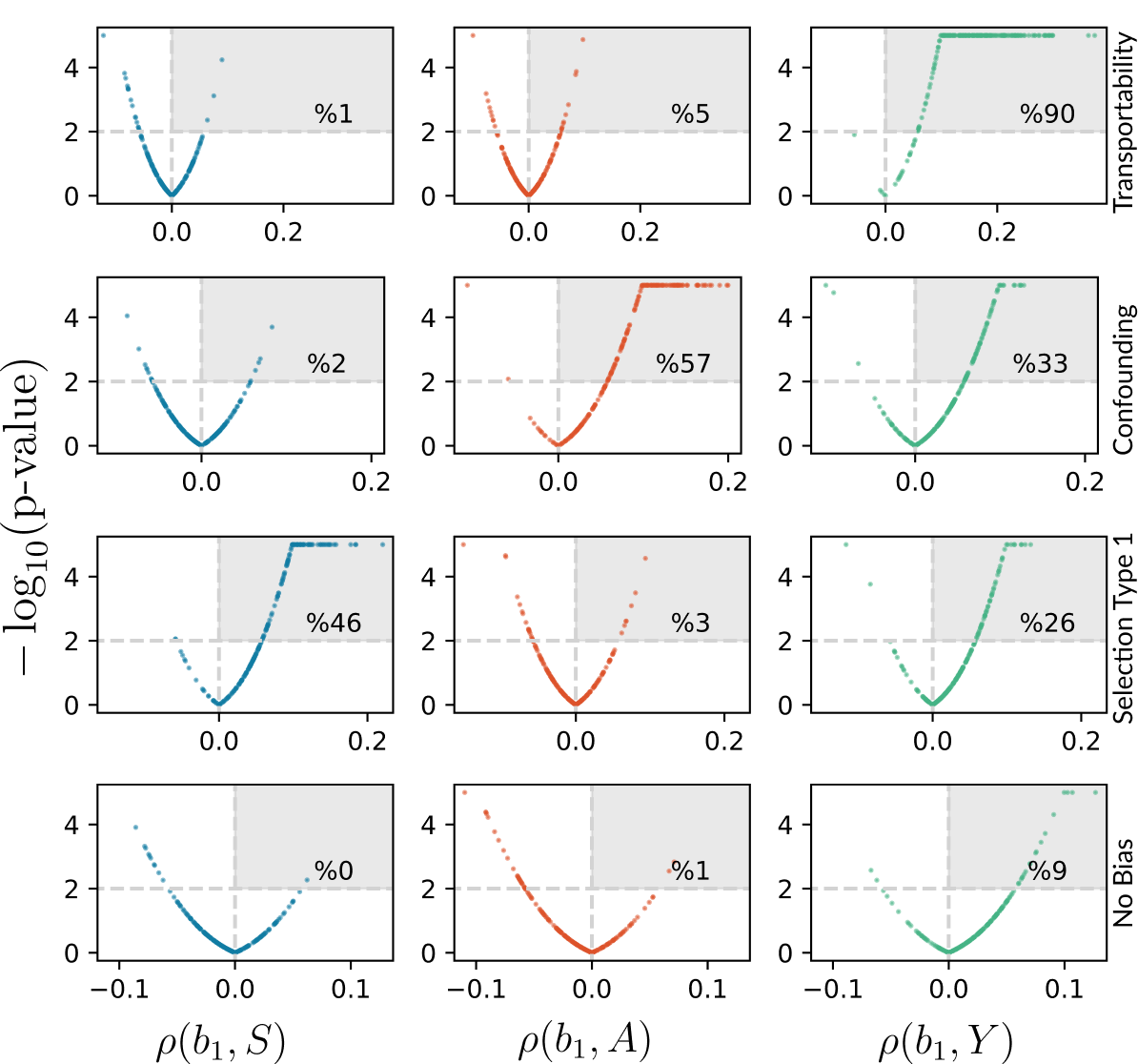}
        \caption{$n_{\text{rct}} = 50000, d=7$}
        \label{fig:50kappd7}
    \end{subfigure}
    \caption{Covariance signals for synthetic experiments with different covariate dimensionalities $d$ and RCT sample size $n_{\text{rct}}$.}
    \label{fig:appsynth}
\end{figure}

\paragraph{Synthetic settings with varying dimensionalities and selection mechanisms:} We conduct the experiments in \Cref{fig:synthetic} with different dimensionalities of the covariates $d \in \{ 5, 7 \}$ ($d=6$ in the main paper) and RCT sample size $n_{\text{rct}} \in \{2000, 50000\}$. The results are presented in \Cref{fig:appsynth}. As dimensionality $d$ increases, the power decreases (and type-1 errors increase), {\em i.e,} it becomes harder to spot the correlation signal statistically significantly. As expected, we observe a similar trend when $n_{\text{rct}}$ decreases, as that leads to larger uncertainty in estimating $\widehat{g}_1(X)$.

We additionally give results for different versions of type 2 selection mechanisms in \Cref{fig:appsynth-sel} (see \Cref{sec:type2_appendix} for analytical results). In general, we find that the sign of the correlations changes with the type of selection mechanism, though is always non-zero.

\paragraph{Synthetic settings with combinations of biases:} In real-world settings, there may be multiple biases present. We experiment with four different combinations of biases: 1) selection bias type 1 and confounding bias, 2) transportability bias and confounding bias, 3) selection bias type 2 and confounding bias, and 4) selection bias type 1 and selection bias type 2. We separately consider the setting with both selection bias type 2 and transportability bias in the next section, \Cref{sec:whi-appendix}. Simulation of each bias in a combination is done as described in the main paper in \Cref{sec:synthetic}. 

\begin{wrapfigure}{r}{0.5\textwidth}
\vspace{-1em} 
\begin{minipage}{0.48\textwidth}
  \hrule
  \hrule
  \vspace{.5em}
  \textbf{Algorithm 2:} Generative Model in the OS \\
  \vspace{-.5em}
  \hrule
  \begin{algorithmic}
        \STATE \textbf{Input:} $\orange{T} \in \{ S,A,Y^0, Y^1 \}$; Boolean \orange{$U$-bias}; \\ Bernoulli Parameter Distribution \orange{${\cal F}$}.
        \hrulefill
        \STATE Let $\orange{p^{T}_{x, u}} \coloneqq P(T\!=\!1 | X\!=\!x, U\!=\!u, R\!=\!0)$
        \FOR{$x \in {\cal X}$}
            \STATE $u \sim \texttt{Unif}(0,1)$ 
            
                \IF{$U$-bias}
                    \STATE $p_1 \sim {\cal F}$, $p_0 \sim {\cal F}$
                    \STATE $p^{T}_{x,u} = u\cdot p_1 + (1-u)\cdot p_0$
                \ELSE
                    \STATE $p^{T}_{x,u} \sim {\cal F}$
                \ENDIF
        \ENDFOR
  \end{algorithmic}
  \hrule
  \hrule
\end{minipage}
\vspace{-1em}
\end{wrapfigure}

Our experiments, shown in~\Cref{fig:syntheticapp}, yield several noteworthy insights. Firstly, when transportability bias is combined with confounding bias, the resulting signals can be a blend of the individual effects. Recall that in our analysis, we observed that with only transportability bias, $\rho(b_1,Y)>0$ and $\rho(b_1,A)=0$, and with only confounding bias, $\rho(b_1,Y)>0$ and $\rho(b_1,A)>0$. Consequently, when both biases are present, we expect $\rho(b_1,A)$ to be somewhat reduced compared to the confounding-only case. This effect is reflected in our results, where   
the percentage of runs showing statistically significant correlations in the biased region drops from $96\%$ (single bias) to $48\%$ (combined biases). 

Secondly, we find that the effect of having only selection bias type 2 on the correlation signals is relatively indistinguishable from the effect when selection bias type 2 is combined with another bias (such as selection bias type 1 or confounding bias).
Finally, when selection bias type 1 is combined with confounding bias, all examined correlations turn positive, i.e. $\rho(b_1,S)>0,\rho(b_1,A)>0,\rho(b_1,Y) >0$. This is a natural result given that there is an unobserved covariate, $U$, influencing $A,Y$ \textit{and} $S$. 

\begin{figure*}[t]
    \centering
    \begin{subfigure}[b]{0.44\linewidth}
        \centering
        \includegraphics[width=\linewidth]{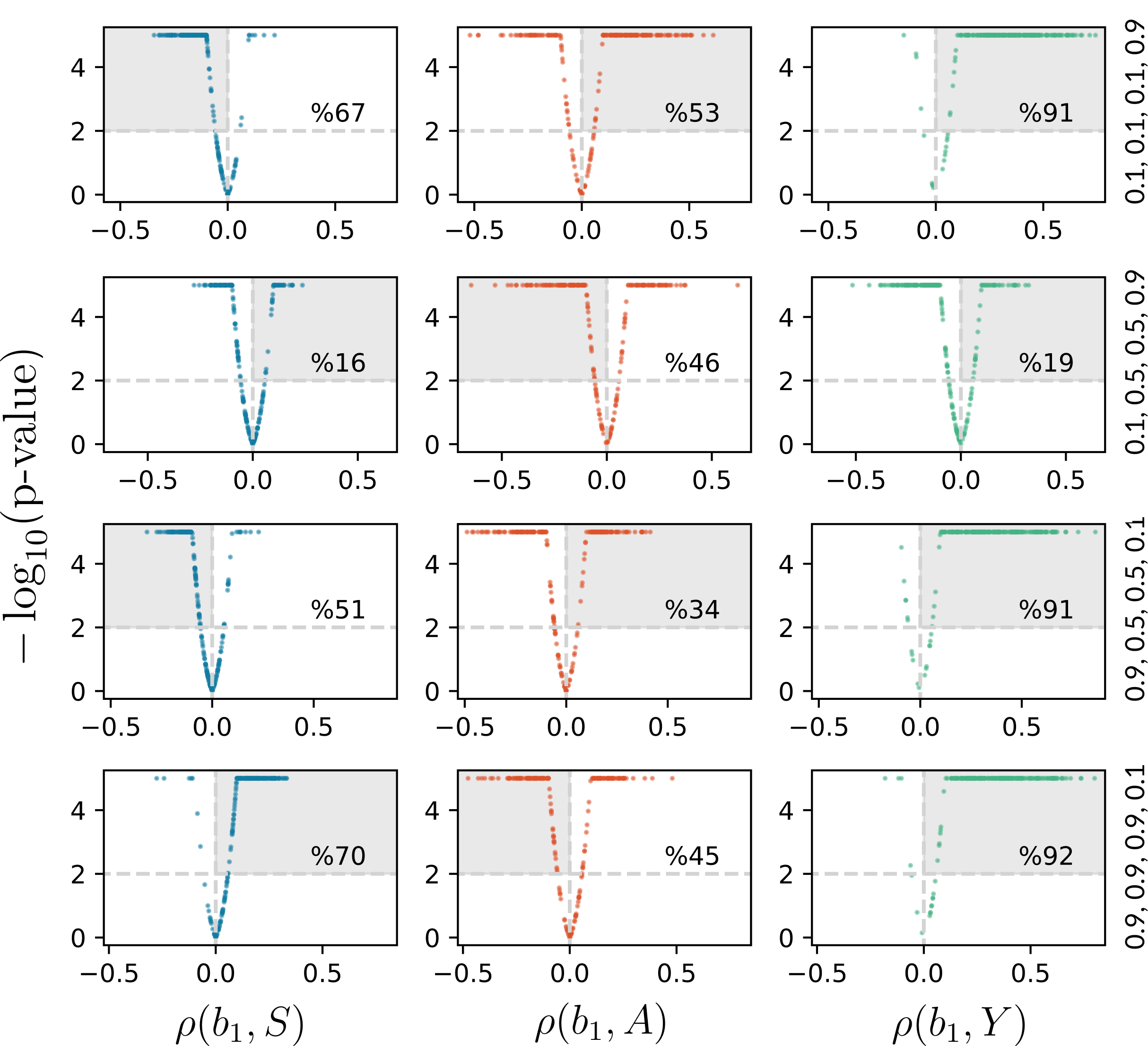}
        \caption{$d=2$}
        \label{fig:d2-select-app}
    \end{subfigure} \hspace{20pt}
    \begin{subfigure}[b]{0.44\linewidth}
        \centering
        \includegraphics[width=\linewidth]{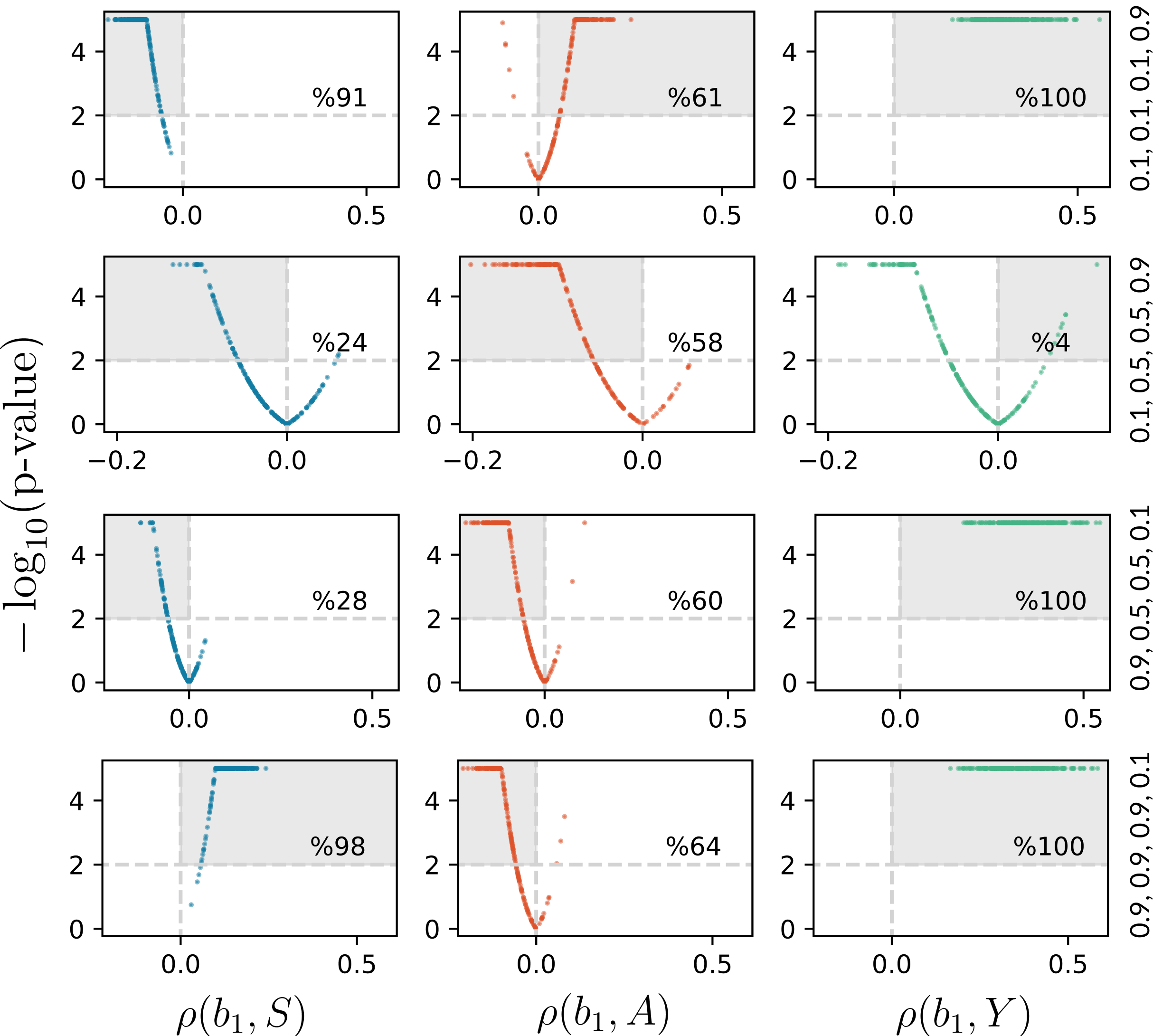}
        \caption{$d=6$}
        \label{fig:d6-select-app}
    \end{subfigure}
    \caption{Covariance signals for synthetic experiments with different selection mechanisms and covariate dimensionalities $d$.}
    \label{fig:appsynth-sel}
\end{figure*}

\paragraph{Synthetic settings with continuous $U$:} In the main paper, we assume a binary, unobserved confounder $U$ that captures residual effects of multiple unmeasured confounders, primarily to simplify the theoretical analysis. Here, we relax this assumption and extend our synthetic experiments to allow for a continuous $U$. The revised generative model is shown in Algorithm 2. Intuitively, when unobserved confounding is present in a patient subgroup (indicated by the $U$-bias flag), the probability of the target variable becomes a convex combination of two samples from $\mathcal{F}$, weighted by $U \sim \texttt{Unif(0,1)}$. 

\begin{figure}[htbp]
    \centering
    \includegraphics[width=.8\linewidth]{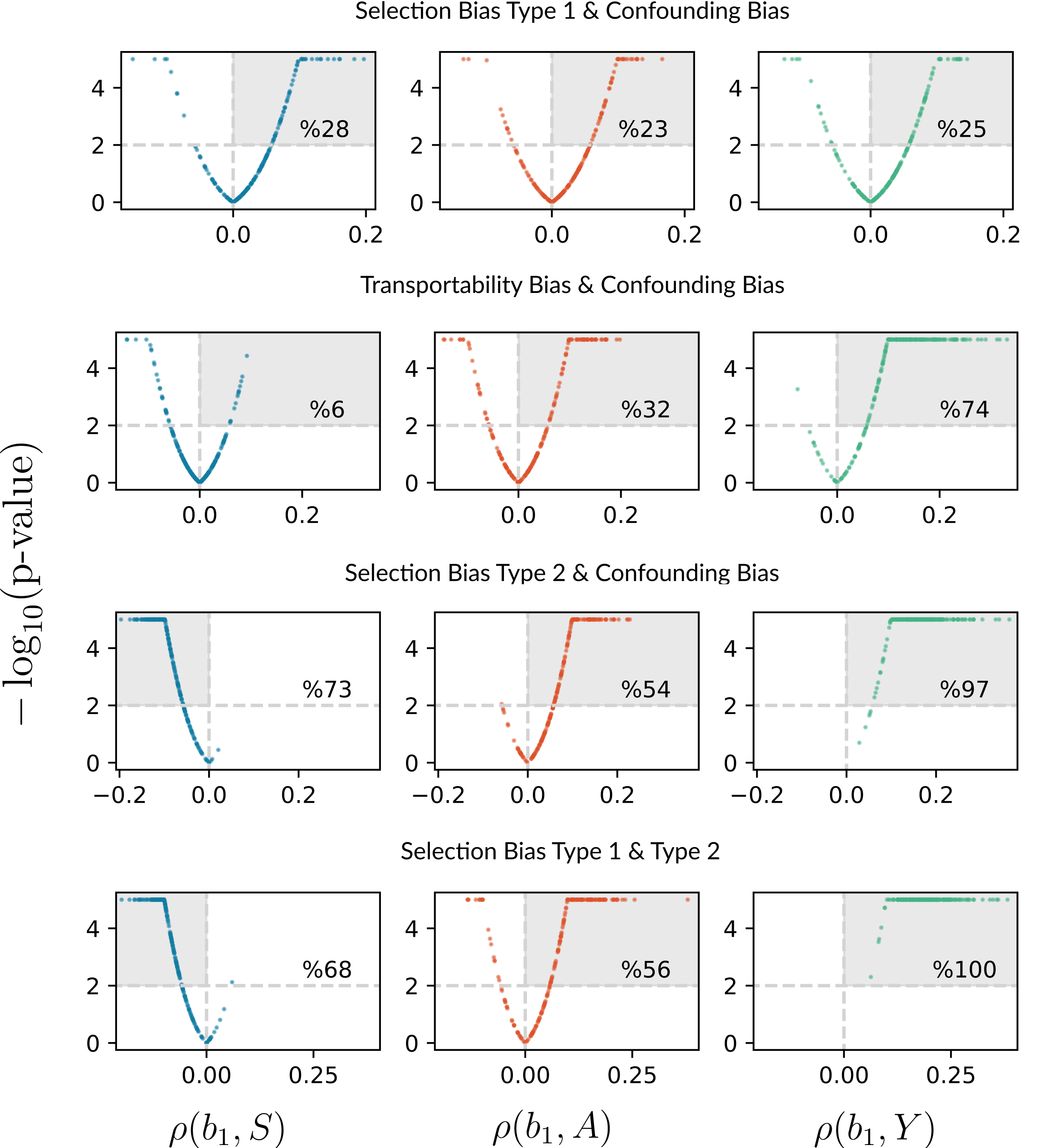}
    \caption{Covariance signals for synthetic settings involving different combinations of biases.}
    \label{fig:syntheticapp}
    \vspace{-10pt}
\end{figure}

\begin{figure}[tbp]
    \centering
    \includegraphics[width=.8\linewidth]{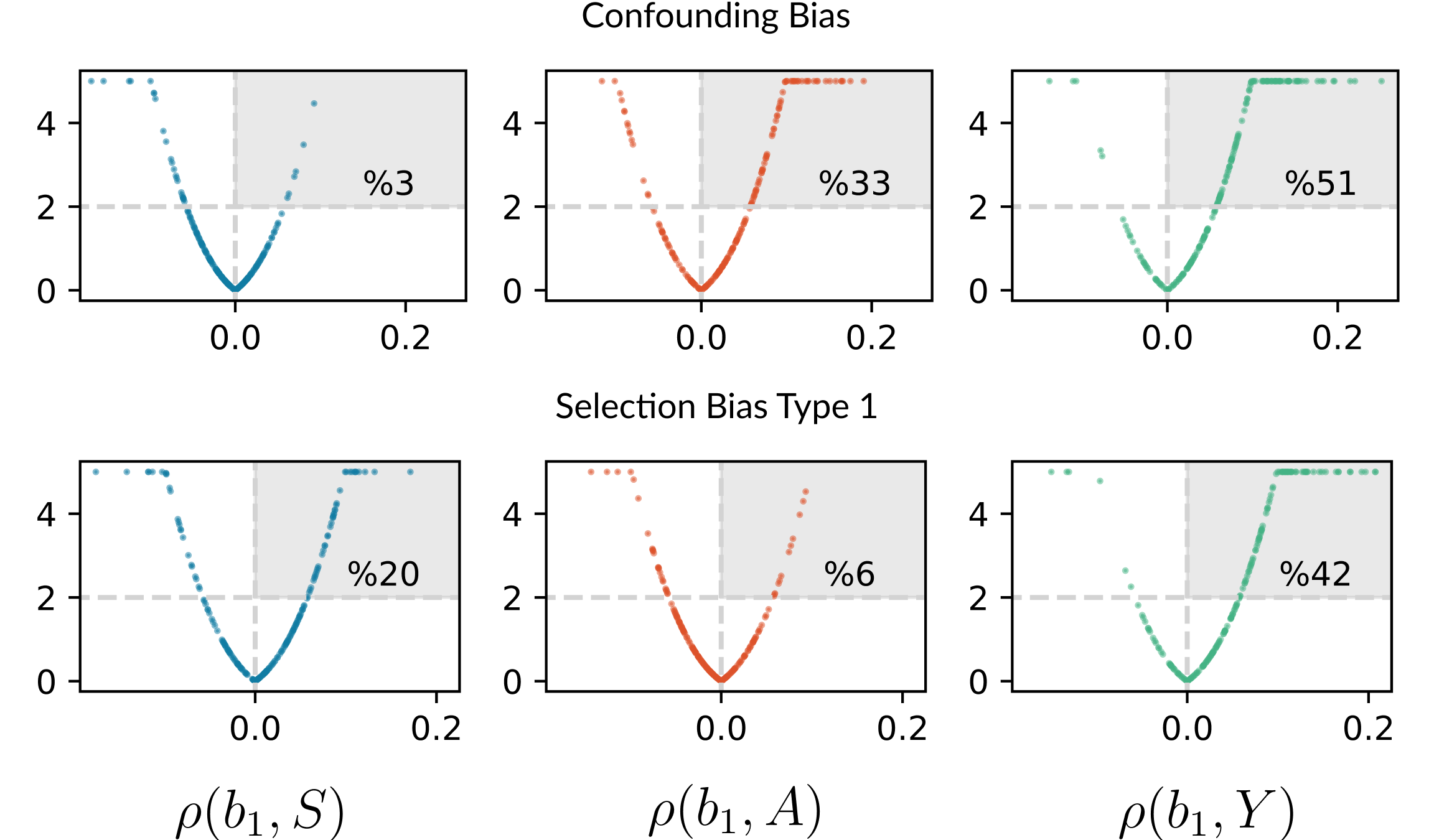}
    \caption{Covariance signals for synthetic setting with continuous $U$.}
    \label{fig:u-continuous}
\end{figure}

We retain the same experimental parameters as in the main paper: $X\in \{0,1\}^d$ with $d=6$, $n_{\text{rct}} = 50000$, $n_{\text{os}} = 50000$, and $n_{\text{val}} = 2000$. The resulting covariance signals, shown in~\Cref{fig:u-continuous}, align with our theoretical predictions, offering further empirical evidence that our framework extends to settings with continuous unobserved confounding.

\subsection{WHI Results}
\label{sec:whi-appendix}

\paragraph{Type 2 selection bias and transportability bias:} To test the hypothesis that residual transportability bias in the WHI experiments yields an amplified $\rho(b_1,Y)$ signal after correcting for the type 2 selection bias, we conduct a synthetic experiment that incorporates both type 2 selection bias and transportability bias to replicate the WHI setting. 

For the type 2 selection bias\footnote{for ease of notation, we denote $P(S=1|Y=y,A=a)$ as $p^S_{ya}$}, we experiment with two different configurations: 1) $p^S_{00} = 0.1, p^S_{01} = 0.1,p^S_{10} = 0.1,p^S_{11} = 0.9$ and 2) $p^S_{00} = 0.9, p^S_{01} = 0.9, p^S_{10} = 0.3, p^S_{11} = 0.1$. One example of the former configuration is Berkson's bias, where people who get exposed to the treatment and have the outcome, e.g. hospitalization, are more likely to be selected. The latter configuration is a rough approximation of the selection mechanism in the WHI OS. Specifically, it is designed to mimic a scenario where patients who are treated and subsequently experience an adverse event are more likely to be filtered out, while patients who either were treated without adverse outcomes or did not receive treatment are preferentially selected.

Following the approach used in our previous experiments, we generate plots that show pairs of Pearson correlation coefficients and their corresponding $p$-values for each run. These plots indicate the percentage of runs where the observed correlation coefficients match the predictions outlined in \Cref{table:bias_table} and \Cref{eq:t2sb_c1} (detailed in \Cref{sec:type2_appendix}), using a significance threshold of $p < 0.01$. We repeat this experiment under the alternative selection probabilities (second row) as well as under conditions designed to minimize selection bias by increasing the selection probabilities to $0.99$ (third row). In all scenarios, we calculate the average correlation for runs falling within the positive or negative region depending on signal, e.g. negative region for $\rho(b_1,A))$, with error bars representing standard deviation across these runs. The complete results are shown in \Cref{fig:seltype2-trans}. 
Importantly, we find that $\rho(b_1,Y)$ tends to increase when minimizing the type 2 selection bias regardless of the selection probabilities, supporting our hypothesis.

\begin{figure}[t]
    \centering
    \includegraphics[width=.9\linewidth]{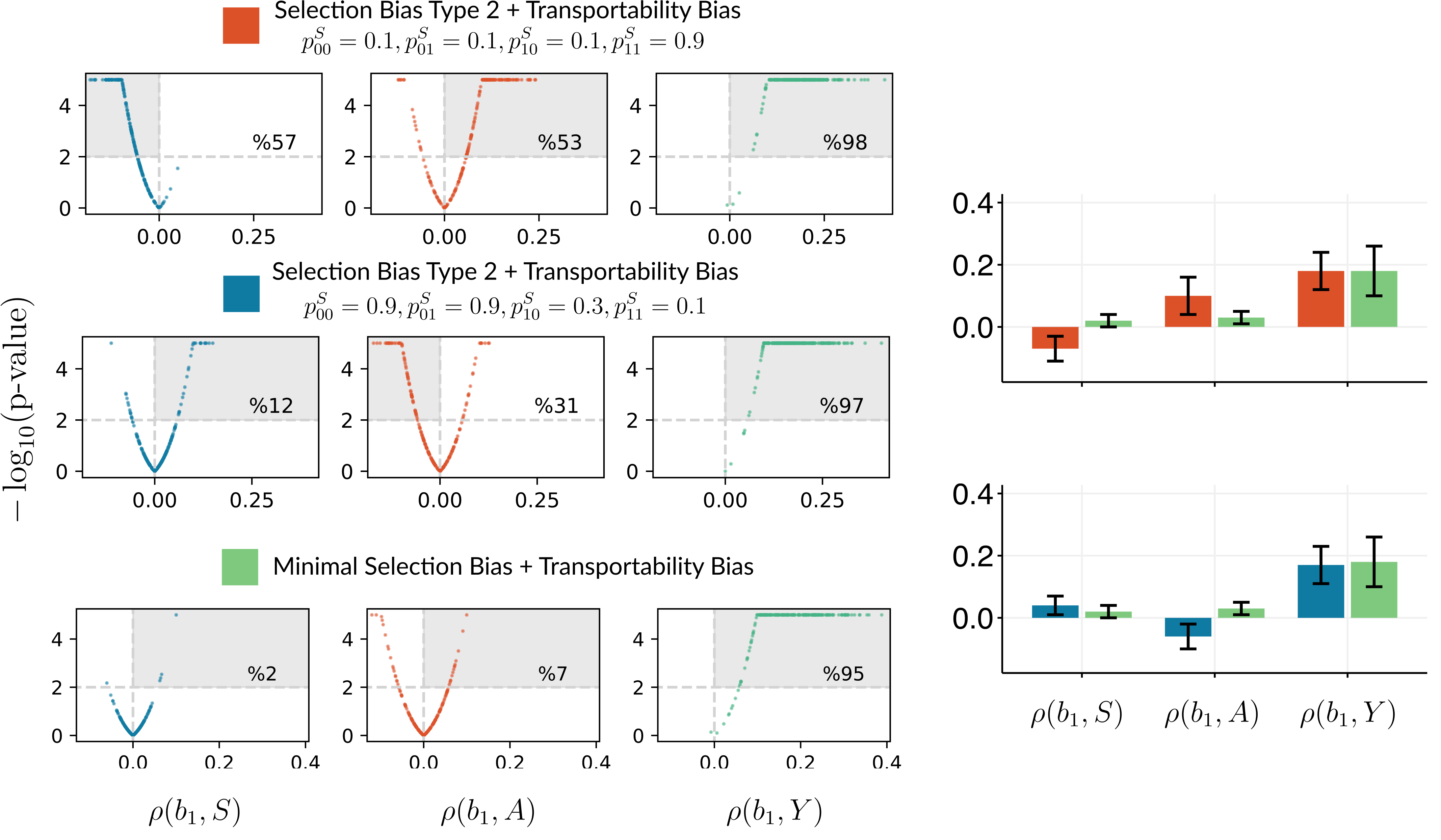}
    \caption{Left: Covariance signals under different synthetic settings: top two have type 2 selection bias (with differing selection probabilities) and transportability bias, bottom has minimal selection bias. Right: Average correlation over runs in positive or negative region depending on signal.}
    \label{fig:seltype2-trans}
    \vspace{-10pt}
\end{figure}

\paragraph{Testing distributional assumption in the WHI data:}
\label{sec:whi-dist}
One important question is how well our theoretical results apply to a real-world dataset like WHI, given that our analysis relies on a specific generative model that may not perfectly capture reality. 
To further justify our generative model, we perform a simple analysis to test our assumption about the underlying sampling distribution in our real-world experiment–specifically the feasibility of modeling the conditional distributions of downstream variables, such as $p(A\mid X,U)$, with a low-uncertainty distribution, $\mathcal{F} (p)\sim \mathcal{F}$ (see \Cref{eq:ddef}).

As shown in~\Cref{fig:test-whi-A}, we plot $P(A \mid X)$ conditioning on different ``sequences'' of subgroups after starting initially with the marginal distribution in the overall population. These subgroups are LLM-generated with a prompt focused towards subgroups that confer an increased or decreased risk for CHD. We use Claude-3.7-Sonnet to generate subgroup sequences. By way of example, the probability of treatment assignment is $33\%$ for the overall WHI observational cohort, relatively close to $50\%$. In the age-based subgroup sequence, the uncertainty in treatment assignment decreases, and the empirical probability of treatment assignment drops. For instance, when first conditioning on $\text{age} > 65$, $P(A=1\mid X)$ decreases from $33\%$ to $22\%$, and then decreases further when conditioning on both $\text{age} > 65$ and elevated BMI. In the case where the subgroups are protective against CHD (bottom row of~\Cref{fig:test-whi-A}, we see that the probability increases towards the ``higher'' region of certainty (i.e. $P(A\mid X) > 0.5$). These patterns show that adding covariates reflecting relevant clinical subgroups reduces uncertainty in treatment assignment, albeit in different directions depending on the baseline risk (for the outcome) in the subgroup.

\begin{figure}[t]
    \centering
    \includegraphics[width=0.85\linewidth]{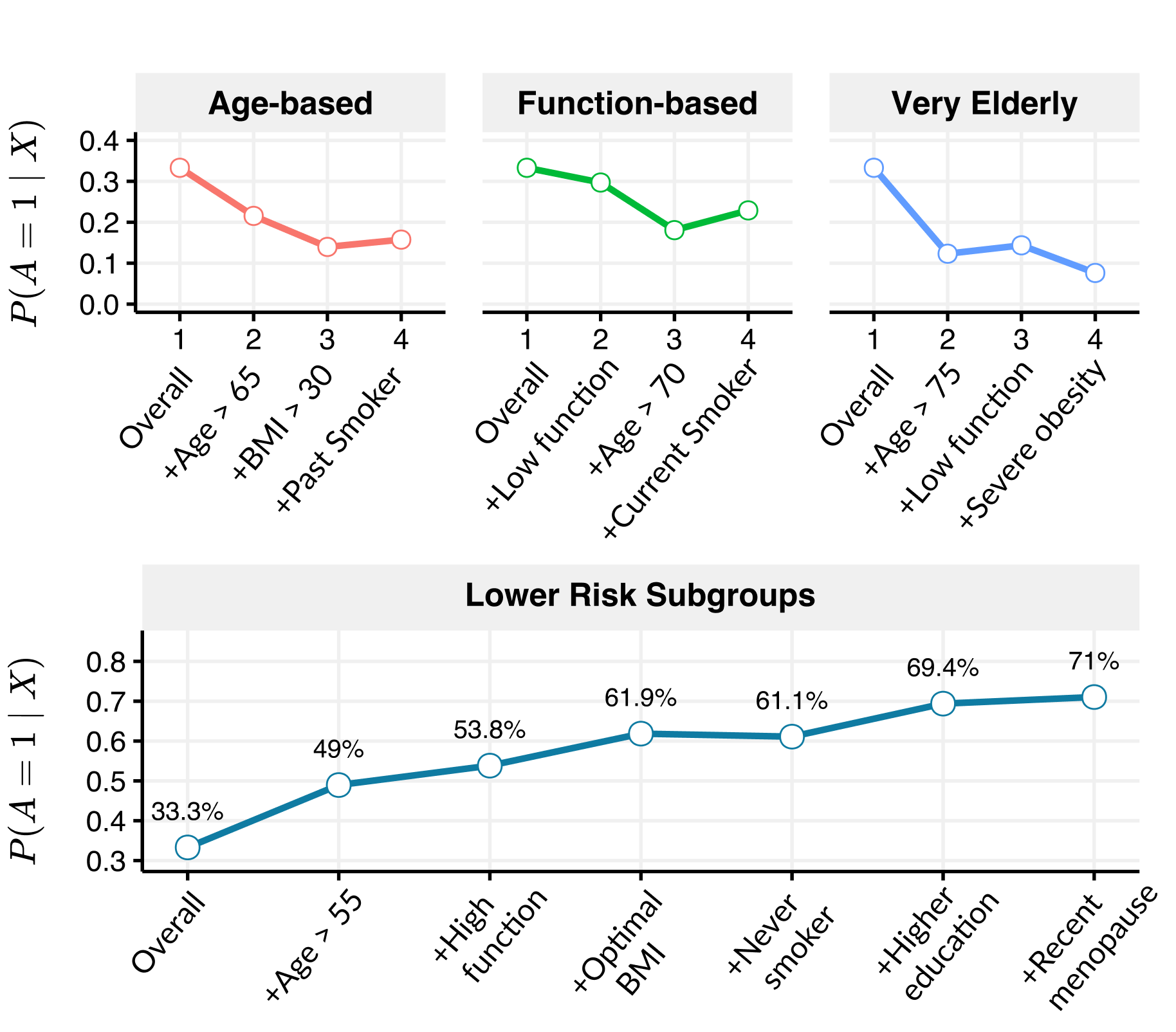}
    \caption{Probability of giving HRT in different sequences of subgroups. On the $x$-axis, we stratify into a specific patient subgroup by adding a constraint on a covariate. Top: Patient subgroups are higher-risk groups for CHD and stroke. Bottom: Patient subgroups are lower-risk groups for CHD and stroke.}
    \label{fig:test-whi-A}
    \vspace{-10pt}
\end{figure}

\paragraph{Additional details for WHI experiments:} In \Cref{tab:whi-replica†ion-results}, we report the average hazard ratios measuring the effect of combined HRT on coronary heart disease (CHD) and stroke outcomes in post-menopausal women. 

\begin{table}[htbp]
  \centering
 \caption{Hazard ratios averaged over population in RCT and OS. We consider ``baseline'' and ``corrected'' settings in OS.}
  \resizebox{\columnwidth}{!}{
  \begin{tabular}{lSSSSSSSS}
    \toprule
    \multirow{2}{*}{Estrogen plus progestin} &
      \multicolumn{2}{c}{\textbf{CHD (Baseline)}} &
      \multicolumn{2}{c}{\textbf{CHD (Corrected)}} &
      \multicolumn{2}{c}{\textbf{Stroke (Baseline)}} & 
      \multicolumn{2}{c}{\textbf{Stroke (Corrected)}}\\
      & HR & \text{95\% CI} & HR & \text{95\% CI} & HR & \text{95\% CI} & HR & \text{95\% CI} \\
      \midrule
    RCT & 1.28 & \text{(1.01,1.61)} & \text{---} & \text{---} & 1.37 & \text{(1.04, 1.8)}  & \text{---} & \text{---} \\
    OS & 0.87 & \text{(0.73,1.03)} & 1.09 & \text{(0.81, 1.45)} & 0.86 & \text{(0.71,1.04)} & 1.29 & \text{(0.95,1.75)} \\
    \bottomrule
  \end{tabular}
  }

  \label{tab:whi-replica†ion-results}
\end{table}

For the nuisance function estimators, we use logistic regression for $\widehat{\eta}_S$ and $\widehat{\eta}_A$, and random forest models for $\widehat{\eta}_Y$ (see \eqref{eq:etas}-\eqref{eq:etay}). We use the default hyperparameters in the scikit-learn implementation of the logistic regression model with maximum iterations of 1000~\citep{sklearn_api}. For the scikit-learn implementation of the random forest, we tune the following hyperparameters on a subset of the training set: 
\begin{itemize}
    \item n-estimators: [100,500]
    \item max-depth: [None, 10, 20]
    \item min-samples-split: [2,10]
\end{itemize}

\section{Related Work} \label{eq:ext_rel_work}
\subsection{Latent variable models}
In the context of latent variable models and graphical models more generally, the main problems, at least relevant to our setting, of interest are parameter estimation. 
Using method of moments for parameter estimation in latent variable models is a common technique. The method of moments paradigm involves 1) computation of certain statistics of the observed data, i.e. means and correlations and 2) finding the model parameters that give rise (approximately) to the same corresponding population quantities. One example of how these operations are executed is through tensor decomposition techniques, as in \citet{anandkumar2014tensor}. At a high level, these methods store low-order moments of the observable data in multidimensional tensors; then, the parameters of the desired model are recovered using tensor decomposition.  \citet{anandkumar2011spectral} assume a causal tree structure, where hidden variables serve as internal nodes and observed variables as leaves. They analyze correlations between observed variables (often through singular values of second moment matrices) to infer relationships through latent variables. See also \citet{ruffini2018new}, who give a nice overview of these kind of methods. 
\subsection{Causal Structure Learning} 
Our work is closely related to the line of work whose goal is to learn a causal structure from observational data and (potentially) experimental data. Our method differs from these approaches in a crucial way. Rather than identifying exact causal DAGs, we focus on categorizing types of bias, where multiple DAGs may be equivalent. There are two broad classes of approaches for structure learning: constraint-based methods and  score-based methods. Below we provide an overview.
 
\paragraph{Score-based methods} The general idea in this line of works is to assign a score to each potential graph, which reflects how well it explains the observed data~\citep{ngscore2022,chobtham2020bayesian}. For instance, \citet{ngscore2022} utilize a scoring function based on the degrees of freedom of a graph (i.e. number of edges in the graph) measuring how well a particular graph explains the empirical covariance matrix  derived from the observed variables.

\paragraph{Constraint-based methods} This line of work uses conditional independence tests of the observed data to infer the underlying causal structure. Some approaches assume causal sufficiency, i.e., assume no selection variables and no unmeasured common causes (e.g., PC, CCD algorithms), while others simply assume the graph is acyclic (e.g., FCI and RFCI algorithms). See \citet{colombo2014order, akbari2021recursive, sadeghi2022conditions} for examples of this approach. The main drawback with these methods in our setting is that the CI-based tests will result in potentially multiple graphs that are Markov equivalent (and there is no way to tell if implied relationships between observed variables are due to e.g. a hidden confounder or selection bias). In contrast, our method can allow practitioners to make this distinction.

\subsection{Distribution Shifts} Our work is similar in spirit to some papers in the literature on distribution shifts but differs in its goals. For instance, \citet{cai2023diagnosing} investigate how much of the decline in prediction performance can be attributed to covariate shift. Similarly, \citet{jin2023diagnosing} attempt to diagnose discrepancies between different studies by breaking them down into components such as sampling variability, observable distribution shifts, and residual factors. However, while their work focuses on quantifying how much of the discrepancy between studies is due to distribution shifts in observed covariates—a specific form of selection bias—we take a broader approach by examining several potential biases, particularly those relevant to the medical domain. 
\subsection{Integrating Evidence from Across Experimental and Observational Studies}
Combining information from RCTs and OSes to develop causally reliable and statistically powerful methods has been of great interest recently. Target trial emulation (TTE), for instance, is a commonly adopted framework to analyze observational data in a principled way \cite{hernan2016using, franklin2021emulating, hernan2022target, wang2023emulation}. 

There is a growing body of work focused on integrating observational and experimental data to enable more reliable, generalizable, and statistically efficient causal inference \cite{bareinboim2016causal, nice2022uk, colnet2024causal, boughdiri2025unified}. Several studies investigate how potentially biased outcome models—learned from observational data—can enhance the statistical power of randomized controlled trials (RCTs) and support the generalization of findings to {\em target} populations or novel experimental settings \cite{kallus2018removing, schuler2022increasing, hatt2022combining, demirel24prediction, cadei2025causal}. Along similar lines, \citet{pmlr-v259-kaul25a} leverage observational data to construct an {\em untrusted prior}, which is then combined with experimental data via {\em conformal prediction} to yield valid and tighter confidence intervals for treatment effects. \citet{guo2022multisource} propose using observational data to construct control variates, reducing variance in ATE estimation within RCTs. Similarly, \citet{wang2025constructing} study the derivation of tighter confidence intervals for the ATE when multiple OSes are available. Lastly, \citet{oberst2023understanding, rosenman2025methods} offer a comprehensive overview of methods that adaptively combine ATE estimates from both experimental and observational data to produce improved hybrid estimators.
\subsection{Benchmarking Observational Studies}
Another line of work focuses on {\em benchmarking}, which involves comparing treatment effect estimates from observational studies (OSes) to those from randomized controlled trials (RCTs) \cite{hartman2015sample, dahabreh2020benchmarking}. Benchmarking methods play a crucial role in assessing the reliability of OSes before they are used in downstream decision-making. Moreover, systematic reviews of benchmarking efforts can offer broader insights into the practical strengths and limitations of OS-based analyses \cite{forbes2020benchmarking, wang2023emulation}.

\citet{de2024detecting,de2024hidden} quantify the bias in an OS by directly comparing its estimates to those from an RCT. Rather than focusing on population-level comparisons, \citet{hussain2022falsification} assess ATE estimates across various subgroups in both studies. Expanding on this approach, \citet{hussain2023falsification, demirel2024benchmarking} compare {\em conditional} ATEs, while enabling {\em automatic} identification of patient subgroups where the OS and RCT estimates diverge. Other work sidesteps the need for individual-level RCT data entirely: \citet{karlsson2024detecting, karlsson2025falsification, xiao2024addressing} show that hidden confounding can be detected—and in some cases mitigated—using multiple OSes under alternative assumptions. Finally, \citet{fawkes2025hardness} establish fundamental limits on the extent to which RCTs can be used to falsify or validate observational findings.
%


\end{document}